\documentclass[10pt,a4paper]{amsart}
\usepackage[utf8]{inputenc}
\usepackage[english]{babel}
\usepackage[mathcal]{euscript}
\usepackage{mathrsfs}  
\usepackage{amsmath}
\usepackage{amsfonts}
\usepackage{amssymb}
\usepackage{ebproof}
\usepackage{stmaryrd}
\usepackage{subcaption}
\usepackage{graphicx}
\usepackage[shortlabels]{enumitem}
\usepackage{proof}
\usepackage{amsthm}
\usepackage{stackengine}
\usepackage{mathtools}
\usepackage{relsize}
\usepackage[all,cmtip]{xy}
\usepackage{tikz-cd}
\usepackage{xfrac}
\usepackage[left=2cm,right=2cm,top=2cm,bottom=2cm]{geometry}
\allowdisplaybreaks
\usepackage{listings}
\usepackage{xcolor}

\definecolor{haskellkeyword}{RGB}{0,0,180}
\definecolor{haskellcomment}{RGB}{0,120,0}
\definecolor{haskellstring}{RGB}{163,21,21}
\definecolor{haskelltype}{RGB}{43,145,175}
\definecolor{haskellbg}{RGB}{248,248,248}

\lstdefinelanguage{Haskell}{
  sensitive=true,
  morekeywords={
    case,class,data,default,deriving,do,else,
    foreign,if,import,in,infix,infixl,infixr,
    instance,let,module,newtype,of,then,type,
    where,qualified,as,hiding,family,forall,
    mdo,rec,proc
  },
  morecomment=[l]--,
  morecomment=[n]{\{-}{-\}},
  morestring=[b]",
  alsoletter={_,'},
}

\usepackage[textsize=tiny,textwidth=46]{todonotes}
\newcommand{\Shuff}{\mathrm{Sf}}

\newcommand{\Syn}{\catfont{Syn}}
\newcommand{\Theo}{\mathscr{T}}

\DeclarePairedDelimiter\norm{\lVert}{\rVert}%

\newcommand{\Quant}{\catfont{V}}
\newcommand{\Cats}{\catfont{Cat}}
\newcommand{\VCat}{\Quant\text{-}\Cats}
\newcommand{\VCatSy}{\Quant\text{-}\Cats_{\mathsf{sym}}}
\newcommand{\VCatSe}{\Quant\text{-}\Cats_{\mathsf{sep}}}
\newcommand{\VCatSS}{\Quant\text{-}\Cats_{\mathsf{sym,sep}}}
\newcommand{\CPTP}{\catfont{CPTP}}
\newcommand{\CPM}{\catfont{CPM}}
\newcommand{\Kar}[1]{\mathcal{K}{(#1)}}

\makeatletter
\DeclareRobustCommand{\iscircle}{\mathord{\mathpalette\is@circle\relax}}
\newcommand\is@circle[2]{%
  \begingroup
  \sbox\z@{\raisebox{\depth}{$\m@th#1\bigcirc$}}%
  \sbox\tw@{$#1\square$}%
  \resizebox{!}{\ht\tw@}{\usebox{\z@}}%
  \endgroup
}
\makeatother

\newcommand{\catfont}[1]{\mathsf{#1}}
\newcommand{\cop}{\catfont{op}}
\newcommand{\fp}{\catfont{fp}}

\newcommand{\catC}{\catfont{C}}

\newcommand{\catD}{\catfont{D}}

\newcommand{\Set}{\catfont{Set}}

\newcommand{\Subs}[2]{\catfont{Sub}_{}}

\newcommand{\Ban}{\catfont{Ban}}

\newcommand{\Met}{\catfont{Met}}

\newcommand{\funfont}[1]{#1}
\newcommand{\funF}{\funfont{F}}

\newcommand{\funG}{\funfont{G}}

\newcommand{\sfunfont}[1]{\mathrm{#1}}

\newcommand{\Id}{\sfunfont{Id}}

\newcommand{\Diag}{\mathscr{D}}

\newcommand{\adjuncte}[2]{%
\begin{tikzcd}[ampersand replacement=\&, column sep=4em]
  \arrow[r, bend left=25, shift left=.8ex, "#2", hook]
  \&
  \arrow[l, bend left=25, shift left=.8ex, "#1"]
  \arrow[from=1-1, to=1-2, phantom, "\top"]
\end{tikzcd}%
}
\newcommand{\adjuncteb}[2]{%
\begin{tikzcd}[ampersand replacement=\&, column sep=4em]
  \arrow[r, bend left=25, shift left=.8ex, "#2"]
  \&
  \arrow[l, bend left=25, shift left=.8ex, "#1",hook]
  \arrow[from=1-1, to=1-2, phantom, "\top"]
\end{tikzcd}%
}

\newcommand{\adjunctopebb}[2]{%
\begin{tikzcd}[ampersand replacement=\&, column sep=4em]
  \arrow[r, bend left=25, shift left=.8ex, "#2"]
  \&
  \arrow[l, bend left=25, shift left=.8ex, "#1", hook]
  \arrow[from=1-1, to=1-2, phantom, "\bot"]
\end{tikzcd}%
}

\newcommand\retract[2]{\xymatrix@=8ex{\ar@{}[r]|{}\ar@<1mm>@/^2mm/@{^{(}->}[r]^{{#2}}
& \ar@<1mm>@/^2mm/@{->>}[l]^{{#1}}}}
\newcommand{\pv}[2]{\langle #1, #2 \rangle}

\DeclareMathOperator{\colim}{\mathrm{colim}}
\newcommand{\inl}{\mathrm{inl}}
\newcommand{\inr}{\mathrm{inr}}

\newcommand{\app}{\mathrm{app}}

\newcommand{\comp}{\cdot}
\newcommand{\id}{\mathrm{id}}
\newcommand{\sw}{\mathrm{sw}}
\newcommand{\spl}{\mathrm{sp}}
\newcommand{\join}{\mathrm{jn}}
\newcommand{\exch}{\mathrm{exch}}
\newcommand{\sh}{\mathrm{sh}}
\newcommand{\dist}{\mathrm{dist}}

\newcommand{\sem}[1]{\left \llbracket #1 \right \rrbracket}

\newcommand{\Nats}{\mathbb{N}}
\newcommand{\Reals}{\mathbb{R}}

\newcommand{\Complex}{\mathbb{C}}

\newcommand{\ie}{\emph{i.e.}}
\newcommand{\eg}{\emph{e.g.}}

\newcommand{\typefont}[1]{\mathbb{#1}}

\newcommand{\typeA}{\typefont{A}}
\newcommand{\typeI}{\typefont{I}}
\newcommand{\typeB}{\typefont{B}}
\newcommand{\typeC}{\typefont{C}}

\newcommand{\rulename}[1]{(\textbf{#1})}

\newcommand{\vljud}{\, \triangleright \, }

\newcommand{\ielim}[2]{#1 \> \> \mathtt{to} \, \ast. \> \> #2}
\newcommand{\telim}[2]{\mathtt{pm} \> \> #1 \> \> \mathtt{to} \> \> x \otimes y. \> \> #2}
\newcommand{\celim}[3]{\mathtt{case} \> \> #1 \> \> \{ \dintrol{x} \rightarrow #2 ; \>\> 
\dintror{y} \rightarrow #3 \}  }
\newcommand{\ccelim}[3]{\mathtt{case} \> \> #1 \> \> \{ \ddintrol{x} \rightarrow #2 ; \>\> 
\ddintror{y} \rightarrow #3 \}  }
\newcommand{\dintrol}[1]{\mathtt{inl}_\typeB (#1)}
\newcommand{\dintror}[1]{\mathtt{inr}_\typeA (#1)}
\newcommand{\ddintrol}[1]{\mathtt{inl} (#1)}
\newcommand{\ddintror}[1]{\mathtt{inr}(#1)}

\theoremstyle{definition}
\newtheorem{definition}{Definition}[section]
\theoremstyle{definition}
\newtheorem{example}[definition]{Example}
\theoremstyle{definition}
\newtheorem{lemma}[definition]{Lemma}
\theoremstyle{definition}
\newtheorem{proposition}[definition]{Proposition}
\theoremstyle{definition}

\theoremstyle{definition}
\newtheorem{theorem}[definition]{Theorem}
\theoremstyle{definition}
\newtheorem{remark}[definition]{Remark}

\author{Renato Neves}
\address{University of Minho \& INESC-TEC, Portugal}
\email{nevrenato@di.uminho.pt}
\author{Bruna Salgado}
\address{University of Minho \& INESC-TEC, Portugal
 and Forschungszentrum Jülich, Germany}
\email{b.salgado@fz-juelich.de}

\thanks{
This work is financed by National Funds through FCT - Fundação para a Ciência e
a Tecnologia, I.P. (Portuguese Foundation for Science and Technology) within
project Banksy, with reference COMPETE2030-FEDER-00892000.
}
\title{Quantalic lambda-calculus and additive disjunction}
\begin{document}

\begin{abstract}
        Motivated by the need to reason about case statements quantitatively,
        we extend quantalic linear $\lambda$-calculus with additive
        disjunction.  We show that the resulting equational system is sound. We
        also show that when certain continuity properties (of the underlying
        quantale) are adopted, it is additionally (approximately) complete. 

        We present several models of the extended calculus, involving for
        example meta-theoretical properties in categorical logic (gluing),
        probabilistic, and quantum computation.  As a concrete application, we
        illustrate how a probabilistic model, based on Banach spaces, can be
        synergistically used with the calculus' equational system to reason
        about Cauchy sequences of random walks.  This highlights the emergent
        shift from ``program semantics as the science of program equivalence"
        to flexible, quantitative perspectives, involving functional analysis
        and beyond.
\end{abstract}

\maketitle

\section{Introduction}
\label{sec:intro}
\subsection{Motivation and context}
Previous work~\cite{dahlqvist22,dahlqvist23} introduced a quantalic
generalisation of linear $\lambda$-calculus, the exponential-free,
\emph{multiplicative fragment} of linear logic. Here we investigate the
incorporation of \emph{additive} structure to this body of work. Our focus is
specifically on the \emph{additive disjunction operator} $\oplus$, which is
typically interpreted via coproducts and gives rise to `case' statements (\ie\
conditionals). Although already interesting from a theoretical point of view,
our motivation for this work stems from practice: in trying to reason
\emph{quantitatively} about (higher-order) programs we often fell short when
these involved conditionals.  Consider for example a single step in a random
walk, written in \textsc{Haskell} style pseudo-code,
\begin{lstlisting}
step :: (Real (*@$\to$@*) Bool) (*@$\to$@*) (Real (*@$\to$@*) Real) 
step p = do x (*@$\leftarrow$ $\mu$@*) ; if p x then jumpLeft else jumpRight
\end{lstlisting}
In words, after receiving a predicate \lstinline{p} the program samples from
distribution $\mu$ and tests the sample against the predicate. Then, depending
on the test result, it returns either action
\lstinline[mathescape]{jumpLeft}
or \lstinline{jumpRight}. In many cases it is useful to substitute $\mu$ by an
approximate distribution and similarly for predicates. For example, if
\lstinline{p} equals $(-)< \frac{1}{2} \sqrt{2}$ it may be better to replace
$\frac{1}{2} \sqrt{2}$ by an approximate rational number for computational
reasons. What is the effect of such a substitution on the overall program?
Similarly real-world scenarios often require knowledge about the effects of
small perturbations on the two different execution branches.  More abstractly,
it is important to ensure that such approximations and/or small perturbations
are congruent, \ie\
\begin{equation}
        \label{eq:cong}
        t \approx s \, \Longrightarrow \, C[t] \approx C[s]
\end{equation}
for all `program contexts' $C[-]$.  But actually, as explained below,
applications involving $\oplus$ are broader than this, and fit in the more
general pattern of reasoning syntactically about co-Cartesian categories
enriched over so-called `\emph{generalised metric
spaces}'~\cite{paseka00,crole93}.

\subsection{Additive disjunction in other works}
A number of important results already considered additive structure in a
quantitative setting.
References~\cite{mardare2016quantitative,mardare2017axiomatizability,mio24,jurka24}
for example are framed in the context of metric (and sometimes ordered)
\emph{universal algebra} and involve additive conjunction (\ie\ $\&$),
typically interpreted via categorical products.  In the more general
higher-order setting, \cite{lago22} enforces \emph{additive conjunction} to be
left adjoint to implication (interpreted via Cartesian-closedness), with a
series of negative results emerging from this.  Our work is orthogonal to these
in that we study the dual of $\&$ (\ie\ $\oplus$) and furthermore assume that
the left adjoint of implication is \emph{multiplicative conjunction} (\ie\
$\otimes$) instead of the additive counterpart. Among other things, this
removes the obstacles discussed in~\cite{lago22}. Of course, we are also
working at the more general level of a quantale, which goes beyond metric and
ordered structures~\cite{paseka00}.

\subsection{Contributions and document structure}
The overarching idea behind quantalic $\lambda$-calculus is that (in)equations
between $\lambda$-terms are labelled by elements $q$ of a quantale. For
example, in the setting of the metric (or Lawvere) quantale, the expression $t
=_q s$ reads as $t$ and $s$ being \emph{at most} at distance $q$ from each
other. The current paper aims at advancing this idea.  Specifically, in
Section~\ref{sec:calculus} we recall quantalic
$\lambda$-calculus~\cite{dahlqvist22,dahlqvist23} and extend it with additive
disjunction. We show that the extension preserves the `unique derivation'
property and the aforementioned congruence property~\eqref{eq:cong}. We also
showcase the extension at work by proving a \emph{quantitative} form of
extensionality for disjunctive types. In Section~\ref{sec:sem} we extend the
calculus' categorical semantics~\cite{dahlqvist22,dahlqvist23} to accomodate
additive disjunction.  This is driven by the following idea: while the
multiplicative structure of the calculus uses an enrichment over the monoidal
structure of a category, the new additive machinery uses an enrichment over the
\emph{Cartesian structure} of the same category.  We show that this extension
is sound and moreover complete if certain \emph{continuity} properties (of the
underlying quantale) are adopted.  We also prove `approximate completeness' in
the absence of the notorious \emph{Archimedean
rule}~\cite{mardare2016quantitative,mardare2017axiomatizability} (it involves
infinitely many premisses and is therefore problematic).

In Section~\ref{sec:models} we present a series of models of our calculus.
Among others, we show that the \emph{gluing construction} extends to the
quantalic setting which gives rise to several ways of studying meta-theoretical
properties of our calculus (and similar ones) categorically
(see~\eg~\cite{lafont88,crole93,hasegawa99,hasegawa99b,hyland03}). Another
interesting case is that of \emph{finite-product preserving} presheaves
enriched over `\emph{generalised metric spaces}': this category is shown to be
a model our calculus by framing it as a \emph{reflective} subcategory of
arbitrary presheaves enriched over the same basis. What is more, this
reflection is built via a first-order variant of our calculus (similar in
spirit to~\cite{rosicky94,fu22} for the \emph{non-enriched} setting). Such a
model is useful not only for studying our calculus' theoretical properties
(\eg\ `concrete completeness' and conservativity results~\cite{johnstone02})
but also to embed co-Cartesian monoidal categories into co-Cartesian monoidal
closed ones in the quantalic setting.

In Section~\ref{sec:prob_quantum} we explore different models of the calculus
in the context of probabilistic~\cite{dahlqvist19,barthe20} and quantum
programming~\cite{selinger2004towards,malherbe13} -- two nascent computational
paradigms. In the former case, we involve the category of Banach spaces and
linear contractions~\cite{dahlqvist19,dahlqvist23}, and illustrate how the
whole framework can be used to reason about Cauchy sequences of (higher-order)
random walks. This highlights the emergent shift from ``program semantics as
the science of program equivalence''~\cite{lago22} to flexible, quantitative
perspectives, involving functional analysis and beyond.  The latter case relies
more heavily on categorical machinery: starting with a certain, enriched
category of \emph{quantum channels} we adopt its \emph{Karoubi
envelope}~\cite{borceux94} to take conditionals into account.  Then in order to
accomodate higher-order structure, we embed this envelope into the category of
finite-product preserving presheaves that was mentioned above.

Finally Section~\ref{sec:conc} concludes and briefly discusses future work. The
paper assumes familiarity with category theory and $\lambda$-calculus.
Section~\ref{sec:prob_quantum} requires basic knowledge of functional analysis
as well. In the paper, we will use the term `co-Cartesian' to refer to
categories with binary coproducts and the term `autonomous' to refer to
symmetric monoidal closed categories.

\section{Quantalic linear lambda-calculus with additive disjunction}
\label{sec:calculus}

\subsection{The term language}
We start by presenting our calculus' term language which is the same as that of
linear $\lambda$-calculus with additive disjuntion~\cite{crole93,maietti05}.
Our presentation will be mostly standard, except for the fact that we will
involve the notion of a \emph{shuffle}~\cite{shulman19,dahlqvist22,dahlqvist23}
(details available below).
First, the calculus' grammar of types is defined by,
\[
  \typeA ::=  X \in G \mid \typeI \mid 
  \typeA \otimes \typeA \mid \typeA \multimap \typeA \mid \typeA \oplus \typeA
\] 
where $G$ is a set of ground types defined \emph{a priori}.  We use Greek
letters $\Gamma,\Delta, E,\dots$ to denote typing contexts, \ie\ lists $x_1 :
\typeA_1,\dots,x_n : \typeA_n$ of typed variables such that each $x_i$ occurs
at most once in $x_1,\dots,x_n$. As already mentioned, our presentation
involves the notion of a shuffle~\cite{shulman19,dahlqvist22,dahlqvist23} --
formally a permutation of typed variables in a context sequence
$\Gamma_1,\dots,\Gamma_n$ such that for all $1 \leq i \leq n$ the relative
order of the variables in $\Gamma_i$ is preserved. For example, if $\Gamma_1 =
x : \typeA, y : \typeB$ and $\Gamma_2 = z : \typeC$ then $z : \typeC, x :
\typeA, y : \typeB$ is a shuffle but $y : \typeB, x : \typeA, z : \typeC$ is
\emph{not} -- because we changed the order in which $x$ and $y$ appear in
$\Gamma_1$.  We denote by $\Shuff(\Gamma_1;\dots;\Gamma_n)$ the set of shuffles
on $\Gamma_1,\dots,\Gamma_n$. As explained
in~\cite{shulman19,dahlqvist22,dahlqvist23b}, the use of shuffles is crucial
for the calculus to have the \emph{exchange} and \emph{unique derivation}
properties at the same time (see also Theorem~\ref{thm:properties} below).

Next, consider a set $\Sigma$ of sorted operation symbols $f :
\typeA_1,\dots,\typeA_n \to \typeA$ with $n \geq 1$. The calculus' judgement
formation rules are presented in Figure~\ref{fig:judg_rules}. In a nutshell,
they conservatively extend those in~\cite{dahlqvist22,dahlqvist23} with
additive disjunction: specifically they now include introduction and
elimination rules that enable the typical case-statement apparatus in
programming (\emph{cf.} dotted line in Figure~\ref{fig:judg_rules}). Note that
context sequences in this figure (\eg\ $\Gamma;\Delta$) appear under the tacit
assumption that different contexts (in a sequence) do not share variables.
Note as well our adoption of Church's typing style.
\begin{figure}[h]
    \begin{equation*}
            \begin{prooftree}
                    \infer0[\rulename{hyp}]{ x : \typeA  \vljud x : \typeA}
            \end{prooftree}
            \hspace{2.5cm}
            \begin{prooftree}
                    \hypo{\Gamma_i \vljud t_i : \typeA}
                    \hypo{f \in \typeA_1, \dots , \typeA_n \to \typeA \in \Sigma}
                    \hypo{E \in \Shuff(\Gamma_1; \dots ; \Gamma_n)}
                    \infer3[\rulename{ax}]{E \vljud f(t_1,\dots,t_n) : \typeA}
            \end{prooftree}
    \end{equation*}
    \vspace{0.3cm}
    \begin{equation*}
            \begin{prooftree}
                    \infer0[\rulename{$\typeI_i$}]{ \vljud \ast : \typeI}
            \end{prooftree}
            \hspace{3.5cm}
            \begin{prooftree}
                    \hypo{\Gamma \vljud t : \typeI}
                    \hypo{\Delta \vljud s : \typeA}
                    \hypo{E \in \Shuff(\Gamma; \Delta)}
                    \infer3[\rulename{$\typeI_e$}]{
                    E \vljud \ielim{t}{s} : \typeA}
            \end{prooftree}
    \end{equation*}
    \vspace{0.3cm}
    \begin{equation*}
            \begin{prooftree}
                    \hypo{\Gamma \vljud t : \typeA}
                    \hypo{\Delta \vljud s : \typeB}
                    \hypo{E \in \Shuff(\Gamma;\Delta)}
                    \infer3[\rulename{$\otimes_i$}]{E \vljud t \otimes s : 
                                \typeA \otimes \typeB}
            \end{prooftree}
            \hspace{0.6cm}
            \begin{prooftree}
                    \hypo{\Gamma \vljud t : \typeA \otimes \typeB}
                    \hypo{\Delta, x : \typeA, y : \typeB \vljud s : \typeC}
                    \hypo{E \in \Shuff(\Gamma; \Delta)}
                    \infer3[\rulename{$\otimes_e$}]{
                    E \vljud \telim{t}{s} : \typeC}
            \end{prooftree}
    \end{equation*}
    \vspace{0.3cm}
    \begin{equation*}
            \begin{prooftree}
                    \hypo{\Gamma, x : \typeA \vljud t : \typeB}
                    \infer1[\rulename{$\multimap_i$}]{
                        \Gamma \vljud \lambda x : 
                        \typeA. \> t : \typeA \multimap \typeB}
            \end{prooftree}
            \hspace{2.5cm}
            \begin{prooftree}
                    \hypo{\Gamma \vljud t : \typeA \multimap \typeB}
                    \hypo{\Delta \vljud s : \typeA}
                    \hypo{E \in \Shuff(\Gamma; \Delta)}
                    \infer3[\rulename{$\multimap_e$}]{
                    E \vljud t \>  s : \typeB}
            \end{prooftree}
    \end{equation*}

    \vspace{0.3cm}
    \dotfill
    \vspace{-0.2cm}
    \begin{equation*}
    \begin{prooftree}
        \hypo{\Gamma \vljud t: \typeA}
        \infer1[\rulename{$\oplus_i$}]{\Gamma \vljud \dintrol{t}: \typeA \oplus \typeB}
    \end{prooftree}
    \hspace{4.5cm}
    \begin{prooftree}
        \hypo{\Gamma \vljud t: \typeB}
        \infer1[\rulename{$\oplus_i$}]{\Gamma \vljud \dintror{t}: \typeA \oplus \typeB}
    \end{prooftree} 
    \end{equation*}
    \vspace{0.3cm}
    \begin{equation*}
    \begin{prooftree}
        \hypo{\Gamma \vljud t: \typeA \oplus \typeB}
        \hypo{\Delta, x: \typeA \vljud s: \typeC}
        \hypo{\Delta, y: \typeB \vljud u: \typeC}
        \hypo{E \in \Shuff(\Gamma; \Delta)}
        \infer4[\rulename{$\oplus_e$}]{E \vljud \celim{t}{s}{u} : \typeC}
    \end{prooftree}
    \end{equation*}
    \caption{Judgement formation rules of quantalic linear $\lambda$-calculus with
    additive disjunction.}
    \label{fig:judg_rules}
\end{figure}

The extended calculus inherits desirable features from the original
calculus~\cite{dahlqvist22,dahlqvist23}. Most notably it inherits the
\emph{unique derivation}, \emph{substitution}, and \emph{exchange} properties,
which rely on the aforementioned shuffling mechanism and Church's typing style.
\begin{theorem}\label{thm:properties}
        The judgement formation rules in Figure~\ref{fig:judg_rules} enjoy the
        following properties.
	\begin{enumerate}
                \item (Unique typing) For any two judgements $\Gamma \vljud t :
                        \typeA$ and $\Gamma \vljud t : \typeA'$, we have
                        $\typeA = \typeA'$; \label{i:unique_type}
                \item (Unique derivation) Every judgement $\Gamma \vljud t :
                        \typeA$ has a unique derivation; \label{i:unique_der}
                \item (Exchange) For every judgement $\Gamma, x : \typeA, y :
                        \typeB, \Delta \vljud t : \typeC$ \label{i:exch} we can
                        derive $\Gamma, y : \typeB, x : \typeA, \Delta \vljud t
                        : \typeC$; 
                \item (Substitution) For every two judgements $\Gamma, x : \typeA
                        \vljud t : \typeB$ and $\Delta \vljud s : \typeA$ we
                        can derive $\Gamma,\Delta \vljud t[s/x] : \typeB$.
                        \label{i:subst} 
	\end{enumerate}
\end{theorem}

\begin{proof}
        A straightforward extension of the analogous proof
        in~\cite{dahlqvist22,dahlqvist23} for the original calculus.
\end{proof}

Henceforth we will often typing information in $\lambda$-terms to keep notation
lightweight.

\subsection{The quantalic (in)equational system}
Classical linear $\lambda$-calculus is equipped with a class of equations
(Figure~\ref{fig:eqs})\footnote{
        The listed equations are standard (see for
        example~\cite{crole93,jacobs99,maietti05}) so we refrain from providing further
        explanations about them. Note however that in order to keep notation
        lightweight, our presentation (in Figure~\ref{fig:eqs}) omits the
        underlying context and typing information, which can be reconstructed
        in the usual way.
},
also known as \emph{equations-in-context} $\Gamma \vljud t = s : \typeA$, that
axiomatise \emph{co-Cartesian autonomous categories} syntactically. The purpose
of the current subsection is to extend such a system to the quantalic setting.
\begin{figure*}[h!]
  \begin{subfigure}[large]{0.6\textwidth}
  \begin{tabular}{r c l}
        $\telim{t \otimes s}{u}$ &
        $=$ & $u[t/x,s/y]$ \\
        $\telim{t}{u[x \otimes y / z]}$ &
        $=$ & $u[t/z]$ \\
        $\ielim{t}{u[\ast/ z]}$ & $=$ & $u[t/z]$ \\
        $\ielim{\ast}{t}$ & $=$ & $t$ 
  \end{tabular}
  \caption{Monoidal structure}
  \end{subfigure}
  \vspace{0.5cm}
  \begin{subfigure}{0.3\textwidth}
  \begin{tabular}{r c l}
        $\lambda x.\ t\ s$ &$=$& $t[s/x]$ \\
        $\lambda x. \> (t\ x)$ &$=$& $t$
  \end{tabular}
  \caption{Higher-order structure}
  \end{subfigure}
  \begin{subfigure}{1\textwidth}
  \centering
  \begin{tabular}{r c l}
        $u[\ielim{t}{w}/z]$ &$=$& $\ielim{t}{u[w/z]}$ \\
        $u[\telim{t}{w}/z]$ &$=$& $\telim{t}{u[w/z]}$
  \end{tabular}
  \caption{Commuting conversions}
  \end{subfigure}
  \begin{subfigure}{1\textwidth}
  \vspace{0.25cm}
  
  \dotfill

  \vspace{0.25cm}
  \centering
  \begin{tabular}{r c l}
          $\ccelim{\ddintrol{t}}{s}{u}$
          & $=$
          &
          $s[t/x]$
          \\
          $\ccelim{\ddintror{t}}{s}{u}$
          & $=$
          &
          $u[t/y]$
          \\
          $\ccelim{t}{s[\ddintrol{x}/z]}{s[\ddintror{y}/z]}$
          & $=$
          &
          $s[t/x]$
  \end{tabular}
  \caption{co-Cartesian structure}
  \end{subfigure}
  \caption{Syntactic axiomatics of co-Cartesian autonomous categories.}
  \label{fig:eqs}
\end{figure*}

Let $\Quant$ denote a commutative and unital quantale~\cite{paseka00,stubbe14},
$\otimes : \Quant \times \Quant \to \Quant$ the corresponding binary operation,
and $k$ the unit.  The following definition is crucial for the quantalic
extension to be \emph{finitary} and at the same time \emph{approximately
complete} (more details about this aspect are given later on).
\begin{definition}[\cite{gierz03}]
	Consider a complete lattice $L$.  For every $x, y \in L$ we say that
	$y$ is \emph{way-below} $x$ (in symbols, $y \ll x$) if for every
	subset $X \subseteq L$ whenever $x \leq \bigvee X$ there exists a
	\emph{finite} subset $F \subseteq X$ such that $y \leq \bigvee F$.
	The lattice $L$ is called \emph{continuous} iff for every $x \in L$,
	\begin{flalign*}
		x = \bigvee \{ y  \mid y \in L\ \text{and} \ y \ll x \}
	\end{flalign*}
        A \emph{basis} $B$ of $L$ is a subset
	$B \subseteq L$ such that for every $x \in L$ the set
	$B \cap \{ y \mid y \in L\ \text{and} \ y \ll x \}$ is directed and
	has $x$ as the least upper bound.
\end{definition}

We will assume that the underlying lattice $L$ of $\Quant$ is continuous and
has a basis $B$ closed under finite joins, binary meets, the unit $k$ and the
multiplication $\otimes$ of the quantale.   
\begin{example}
        \label{ex:quant2}
        The Boolean quantale $(L:=(\{0 \leq 1\}, \vee), \otimes := \wedge, k :=
        \top)$ is finite and thus satisfies the conditions
        above~\cite{gierz03}. More generally, every \emph{coherent locale} satisfies
        these conditions~\cite{johnstone82} when the adopted basis is that of
        \emph{compact} elements (those $x \in L$ such that $x \ll x$). For such
        locales, our formulation of approximate completeness (below)
        instantiates to the usual formulation of completeness.    
        Take now the metric quantale (also known as Lawvere quantale)
        $(L:=([0,\infty], \wedge), \otimes := +, k :=0)$. The way-below
        relation corresponds to the strictly greater relation with $\infty >
        \infty$. A basis for the underlying lattice that satisfies the
        conditions above is the set of extended non-negative rational numbers.
        The latter is also a a basis for the \emph{ultrametric} quantale
        $(L:=([0,\infty], \wedge), \otimes := \max, k:=0)$.  Note that the
        underlying order $\leq$ of these last two quantales is the relation
        $\geq_{[0,\infty]}$ rather than $\leq_{[0,\infty]}$. Thus in this
        setting the top element is actually $0$. 
\end{example}

We now present the quantalic extension of Figure~\ref{fig:eqs}. Recall from the
introduction that the idea of quantalic $\lambda$-calculus boils down to
labelling (in)equations by elements of a quantale $\Quant$. And in particular,
a classical equation $t = s$ now abbreviates $t \leq_k s$ together with $s
\leq_k t$.  Thus given a quantale $\Quant$ that satisfies the aforementioned
constraints, the corresponding quantalic system adds to Figure~\ref{fig:eqs}
the laws in Figure~\ref{fig:equations-in-context-cond}.  The first four rules
are essentially a quantalic generalisation of classical (in)equational laws.
The other rules express a form of \emph{quantalic congruence}.  In particular,
the ones below the dotted line are the ones that are new
w.r.t.~\cite{dahlqvist22,dahlqvist23} and dictate how the contructs relative to
disjunctive types interact with the quantalic structure.  Most notably, the
expression $p \otimes (q \wedge r)$ between case statements encodes a form of
\emph{worst-case} assumption: intuitively one takes the `worst' value from
$\{q, r\}$ to reflect the possibility of taking the branch in which the two
respective terms `differ' the most -- such a value then compounds with $p$ to
reflect the `difference' between the tests $t$ and $s$. In the metric setting,
for example, $p \otimes (q \wedge r)$ instantiates to $p + (q \vee r)$. 

Observe that while the original congruence laws~\cite{dahlqvist22,dahlqvist23}
only make use of the quantale's linear structure (\ie\ $\otimes$), the extended
version (Figure~\ref{fig:equations-in-context-cond}) additionally resorts to
the quantale's Cartesian structure (\ie\ infima).  This ties up nicely with the
corresponding categorical semantics, which we detail in the following section.
Note as well that some of the presented laws are redundant -- for example the injection
laws are subsumed by the substitution law. We have decided however to keep this
redundancy to make more explicit how the available term constructs interact
with the quantalic structure. Moreover note that the substitution rule entails
the implication,
\[
        t \leq_q s \, \Longrightarrow \, u[t/x] \leq_q u[s/x]
\]
This formalises the congruence property~\eqref{eq:cong} that was alluded to in
the introduction, with $u$ representing a `program context'. Finally observe
that rule \rulename{join} entails $t \leq_\bot s$ for any two terms $t$ and
$s$. While looking rather trivial, this is essential for our approximate
completeness result.
\begin{remark}
        We omitted the controversial Archimedean
        rule~\cite{dahlqvist22,dahlqvist23,mardare2016quantitative} from our
        system, due to its use of an infinite number of premisses. As already
        mentioned, the only `cost' is that the usual completeness result must
        be weakened to `approximate completeness'.  Note as well that the
        system in Figure~\ref{fig:equations-in-context-cond} does not include a
        symmetry law,
        \[
                \infer[\rulename{sym}]{s \leq_q t}{t \leq_q s}
        \]
        despite being desirable to have it in some cases (\eg\ for
        (ultra)metric reasoning). To accomodate this, we will use the notation
        $t =_q s$ (rather than $t \leq_q s$) to tacitly postulate the symmetry
        law when needed.
\end{remark}

  \begin{figure}[h!]
    \begin{equation*}
            \begin{prooftree}
                    \hypo{}
                    \infer1[\rulename{refl}]{t \leq_k t}
            \end{prooftree}
            \hspace{1cm}
            \begin{prooftree}
                    \hypo{ t \leq_q s }
                    \hypo{ s \leq_r u }
                    \infer2[\rulename{trans}]{t \leq_{q \otimes r} u}
            \end{prooftree}
            \hspace{1cm}
            \begin{prooftree}
                    \hypo{ t \leq_q s }
                    \hypo{ r \leq q}
                    \infer2[\rulename{weak}]{t \leq_r s}
            \end{prooftree}
            \hspace{1cm}
            \begin{prooftree}
                    \hypo{ \forall i \leq n.  \, t \leq_{q_i} s}
                    \infer1[\rulename{join}]{ t \leq_{\vee q_i} s }
            \end{prooftree}
    \end{equation*}
    \vspace{0.3cm}
    \begin{equation*}
            \begin{prooftree}
                    \hypo{t \leq_q s}
                    \hypo{u \leq_r v}
                    \infer2[\rulename{subst}]{ t[u/x] \leq_{q \otimes r} s[v/x] }
            \end{prooftree}
            \hspace{5cm}
            \begin{prooftree}
                    \hypo{ \Gamma \vljud t \leq_q s : \typeA }
                    \hypo{ \Gamma \simeq_\pi \Delta }
                    \infer2[\rulename{perm}]{ \Delta \vljud t \leq_q s : \typeA}
            \end{prooftree}
    \end{equation*}
    \vspace{0.3cm}
    \begin{equation*}
            \begin{prooftree}
                    \hypo{\forall i \leq n. \> t_i \leq_{q_i} s_i}
                    \infer1[]{f(t_1,\dots,t_n) \leq_{\otimes q_i} f(s_1,\dots,s_n)}
            \end{prooftree}
            \hspace{2cm}
            \begin{prooftree}
                    \hypo{ t \leq_q s }
                    \hypo{ u \leq_r v }
                    \infer2[]{ t \otimes u \leq_{q \otimes r} s \otimes v}
            \end{prooftree}
            \hspace{2cm}
            \begin{prooftree}
                    \hypo{ t \leq_q s }
                    \hypo{ u \leq_r v }
                    \infer2[]{\ielim{t}{u} \leq_{q \otimes r} \ielim{s}{v}}
            \end{prooftree}
     \end{equation*}
     \vspace{0.3cm}
     \begin{equation*}  
            \begin{prooftree}
                    \hypo{ t \leq_q s }
                    \hypo{ u \leq_r v }
                    \infer2[]{\telim{t}{u} \leq_{q \otimes r} \telim{s}{v} }
            \end{prooftree}
            \hspace{2cm}
            \begin{prooftree}
                    \hypo{ t \leq_q s }
                    \infer1[]{\lambda x. \, t \leq_q \lambda x. \, s}
            \end{prooftree}
            \hspace{2cm}
            \begin{prooftree}
                    \hypo{ t \leq_q s }
                    \hypo{ u \leq_r v }
                    \infer2[]{t \> u \leq_{q \otimes r} s \> v}
            \end{prooftree}
      \end{equation*}
     \vspace{0.3cm}
     \vspace{0cm}
  
     \dotfill

     \vspace{0cm}
    \begin{equation*}
            \begin{prooftree}
                    \hypo{ t \leq_q s}
                    \infer1[]{\ddintrol{t} \leq_q \ddintrol{s}}
            \end{prooftree}
            \hspace{5cm}
            \begin{prooftree}
                    \hypo{ t \leq_q s}
                    \infer1[]{\ddintror{t} \leq_q \ddintror{s}}
            \end{prooftree}
    \end{equation*}
    \vspace{0.3cm}
    \begin{equation*}
            \begin{prooftree}
                    \hypo{ t \leq_p s }
                    \hypo{ u \leq_q v }
                    \hypo{ w \leq_r o }
                    \infer3[]{\ccelim{t}{u}{w} \leq_{p \otimes (q \wedge r)} \ccelim{s}{v}{o}}
            \end{prooftree}
    \end{equation*}
    \caption{The quantalic (in)equational system.}
    \label{fig:equations-in-context-cond}
    \end{figure}

\begin{example}
        We conclude the section by briefly illustrating quantitative reasoning
        via the quantalic system that was just presented. Specifically we
        involve the notion of extensionality relative to disjunctive types: in
        the classical setting, \emph{strict} equality between two
        (higher-order) terms $t,s : \typeA_1 \oplus \typeA_2 \multimap \typeB$
        is entailed by the two equalities below.
        \[
                \begin{cases}
                        x : \typeA_1 \vljud t \> (\ddintrol{x}) = s \> (\ddintrol{x}) 
                        \\
                        y : \typeA_2 \vljud t \> (\ddintror{y}) = s \> (\ddintror{y})
                \end{cases}
        \]
        In words, this says that equality between $t$ and $s$ (whose input type
        is disjunctive) follows from establishing equality between the two
        terms restricted to inputs of type $\typeA_1$ and then $\typeA_2$.
        Now, there is a quantitative analogue, concretely the quantalic
        inequations,
        \[
                \begin{cases}
                        x : \typeA_1 \vljud t \> (\ddintrol{x}) \leq_q s \> (\ddintrol{x})
                        \\
                        y : \typeA_2 \vljud t \> (\ddintror{y}) \leq_r s \> (\ddintror{y})
                \end{cases}
        \]
        entail $t \leq_{q \wedge r} s$. Such is proved in two stages: first, an
        application of the last rule in
        Figure~\ref{fig:equations-in-context-cond} and an appeal to the
        extensionality law of case statements (in Figure~\ref{fig:eqs}) yields
        the equation $t \> z \leq_{q \wedge r} s \> z$.  Second, congruence
        relative to $\lambda$-abstraction
        (Figure~\ref{fig:equations-in-context-cond}) and the latter's
        extensionality (Figure~\ref{fig:eqs}) yields $t \leq_{q \wedge r} s$.
\end{example}

\section{Categorical semantics}
\label{sec:sem}
\subsection{Interpretation of judgements}
We start by recalling the interpretation of judgements $\Gamma \vljud t:
\typeA$ in a co-Cartesian autonomous category $\catC$. First, for all
$\catC$-objects $X,Y,Z$, $\sw_{X,Y} : X \otimes Y \to Y \otimes X$ denotes the
symmetry morphism, $\lambda_X : I \otimes X \to X$ the left unitor,   and
$\alpha_{X,Y,Z} : X \otimes (Y \otimes Z) \to (X \otimes Y) \otimes Z$ the left
associator. For all $\catC$-morphisms $f : X \otimes Y \to Z$ the morphism
$\overline{f} : X \to (Y \multimap Z)$ denotes the corresponding curried
version (\ie\ the right transpose). The two injections into a binary coproduct
are denoted by $\inl : X \to X + Y$ and $\inr : Y \to X + Y$. Recall that all
co-Cartesian autonomous categories are distributive, \ie\ $[\id \otimes \inl ,
\id \otimes \inr] : X \otimes Y + X \otimes Z \to X \otimes (Y + Z)$ is an
isomorphism. We denote its inverse by $\dist_{X,Y,Z}$.  Given $X_1,\dots,X_n
\in \catC$ we write $X_1 \otimes \dots \otimes X_n$ for the $n$-tensor $(\dots
(X_1 \otimes X_2) \otimes \dots ) \otimes X_n$, and similarly for
$\catC$-morphisms. We will often omit subscripts in components of natural
transformations if no ambiguities arise.

For every ground type $X \in G$ we postulate an interpretation $\sem{X}$
as a $\catC$-object. The remaining types are interpreted as usual: inductively
using the co-Cartesian autonomous structure of $\catC$.  Given a non-empty
context $\Gamma = \Gamma', x : \typeA$, its interpretation is defined by
$\sem{\Gamma', x : \typeA} = \sem{\Gamma'} \otimes \sem{\typeA}$ if $\Gamma'$
is non-empty and $\sem{\Gamma', x : \typeA} = \sem{\typeA}$ otherwise. The
empty context is defined by $\sem{-} = I$.  We will involve some housekeeping
morphisms to handle interactions between context interpretation and the
autonomous structure of $\catC$. Specifically, given $\Gamma_1,\dots,\Gamma_n$
we denote by $\spl_{\Gamma_1; \dots;\Gamma_n} : \sem{\Gamma_1, \dots, \Gamma_n}
\to \sem{\Gamma_1} \otimes \dots \otimes \sem{\Gamma_n}$ the morphism that
splits $\sem{\Gamma_1, \dots, \Gamma_n}$ into $\sem{\Gamma_1} \otimes \dots
\otimes \sem{\Gamma_n}$, and its inverse by $\join_{\Gamma_1;\dots;\Gamma_n}$.
Given $\Gamma, x : \typeA, y : \typeB, \Delta$ we denote by $\exch_{\Gamma,
\underline{x : \typeA, y : \typeB}, \Delta} : \sem{\Gamma, x : \typeA, y :
\typeB, \Delta} \to \sem{\Gamma, y : \typeB, x : \typeA, \Delta}$ the morphism
that permutes $x$ with $y$.  We will also resort to `shuffling' morphisms
$\sh_E : \sem{E} \to \sem{\Gamma_1,\dots,\Gamma_n}$ which, as described in
Section~\ref{sec:calculus}, correspond to restricted permutation sequences. The
full definition of all these housekeeping morphisms is presented
in~\cite{dahlqvist22,dahlqvist23}. 

Finally let us postulate an interpretation $\sem{f} : \sem{\typeA_1} \otimes
\dots \otimes \sem{\typeA_n} \to \sem{\typeA}$ for every operation symbol $f :
\typeA_1,\dots,\typeA_n \to \typeA$ in $\Sigma$. The rules for interpreting
judgements $\Gamma \vljud t : \typeA$ are then given in
Figure~\ref{fig:typing_rules_cond}. Since these are quite standard we omit
explanations, as well as the proofs of the following results.
\begin{figure}[h]
    \begin{equation*}
            \begin{prooftree}
                    \infer0[]{\sem{x : \typeA  \vljud x : \typeA} 
                    = \id}
            \end{prooftree}
            \hspace{2.5cm}
            \begin{prooftree}
                    \hypo{\sem{\Gamma_i \vljud t_i : \typeA} = m_i}
                    \hypo{f \in \typeA_1, \dots , \typeA_n \to \typeA \in \Sigma}
                    \hypo{E \in \Shuff(\Gamma_1; \dots ; \Gamma_n)}
                    \infer3[]{\sem{E \vljud f(t_1,\dots,t_n)} 
                        = \sem{f} \comp (m_1 \otimes \dots \otimes m_n) 
                \comp \spl_{\Gamma_1; \dots ; \Gamma_n} \comp \sh_E}
            \end{prooftree}
    \end{equation*}
    \vspace{0.3cm}
    \begin{equation*}
            \begin{prooftree}
                    \infer0[]{\sem{\vljud \ast : \typeI} = \id}
            \end{prooftree}
            \hspace{3.5cm}
            \begin{prooftree}
                    \hypo{\sem{\Gamma \vljud t : \typeI} = m}
                    \hypo{\sem{\Delta \vljud s : \typeA} = n}
                    \hypo{E \in \Shuff(\Gamma; \Delta)}
                    \infer3[]{
                    \sem{E \vljud \ielim{t}{s} : \typeA} = 
                        n \comp \lambda \comp (m \otimes \id)
                        \comp \spl_{\Gamma;\Delta} \comp \sh_E}
            \end{prooftree}
    \end{equation*}
    \vspace{0.3cm}
    \begin{equation*}
            \begin{prooftree}
                    \hypo{\sem{\Gamma \vljud t : \typeA} = m}
                    \hypo{\sem{\Delta \vljud s : \typeB} = n}
                    \hypo{E \in \Shuff(\Gamma;\Delta)}
                    \infer3[]{\sem{ E \vljud t \otimes s : 
                    \typeA \otimes \typeB} = (m \otimes n) \comp 
                    \spl_{\Gamma;\Delta} \comp \sh_E}
            \end{prooftree}
     \end{equation*}
     \vspace{0.3cm}
     \begin{equation*}
            \begin{prooftree}
                    \hypo{\sem{\Gamma \vljud t : \typeA \otimes \typeB} = m}
                    \hypo{\sem{\Delta, x : \typeA, y : \typeB \vljud s : \typeC} 
                    = n}
                    \hypo{E \in \Shuff(\Gamma; \Delta)}
                    \infer3[]{
                    \sem{E \vljud \telim{t}{s} : \typeA}
                        = n \comp \join_{\Delta;\typeA;\typeB} \comp \alpha
                        \comp \sw \comp (m \otimes \id) \comp \spl_{\Gamma;\Delta}
                        \comp \sh_E }
            \end{prooftree}
    \end{equation*}
    \vspace{0.3cm}
    \begin{equation*}
            \begin{prooftree}
                    \hypo{\sem{\Gamma, x : \typeA \vljud t : \typeB} = m}
                    \infer1[]{
                        \sem{\Gamma \vljud \lambda x. \> t : \typeA \multimap \typeB}
                        = \overline{m \comp \join_{\Gamma;\typeA}}
                        }
            \end{prooftree}
            \hspace{2.5cm}
            \begin{prooftree}
                    \hypo{\sem{\Gamma \vljud t : \typeA \multimap \typeB} =m}
                    \hypo{\sem{\Delta \vljud s : \typeA} = n}
                    \hypo{E \in \Shuff(\Gamma; \Delta)}
                    \infer3[]{
                    \sem{E \vljud t \>  s : \typeB} 
                        = \app \comp (m \otimes n) \comp \spl_{\Gamma;\Delta}
                        \comp \sh_E }
            \end{prooftree}
    \end{equation*}

    \vspace{0.3cm}
    \dotfill
    \vspace{-0.2cm}
    \begin{equation*}
    \begin{prooftree}
        \hypo{\sem{\Gamma \vljud t: \typeA} = m}
        \infer1[]{\sem{\Gamma \vljud \ddintrol{t}: \typeA \oplus \typeB}
                = \inl \comp m}
    \end{prooftree}
    \hspace{4.5cm}
    \begin{prooftree}
        \hypo{\sem{\Gamma \vljud t: \typeB} = m}
        \infer1[]{\sem{\Gamma \vljud \ddintror{t}: \typeA \oplus \typeB}
                = \inr \comp m}
    \end{prooftree} 
    \end{equation*}
    \vspace{0.3cm}
    \begin{equation*}
    \begin{prooftree}
        \hypo{\sem{\Gamma \vljud t: \typeA \oplus \typeB} = m}
        \hypo{\sem{\Delta, x: \typeA \vljud s: \typeC} = n}
        \hypo{\sem{\Delta, y: \typeB \vljud u: \typeC} = o}
        \hypo{E \in \Shuff(\Gamma; \Delta)}
        \infer4[]{\sem{E \vljud \ccelim{t}{s}{u} : \typeC} =
                [n \comp \join_{\Delta;\typeA} \comp \sw,
                o \comp \join_{\Delta;\typeB} \comp \sw] \comp 
                \dist \comp
                (m \otimes \id) \comp \spl_{\Gamma;\Delta} \comp \sh_E}
    \end{prooftree}
    \end{equation*}
    \caption{Interpretation rules of quantalic linear $\lambda$-calculus with 
    additive disjunction.}
    \label{fig:typing_rules_cond}
\end{figure}
\begin{lemma}[Exchange and Substitution]
  \label{lem:exch_subst_inter}
  For all judgements $\Gamma, x : \typeA, y : \typeB, \Delta \vljud t :
  \typeC$, $\> \Gamma, x : \typeA \vljud s : \typeB$, and $\Delta \vljud w :
  \typeA$, the following equations hold in every co-Cartesian autonomous
  category $\catC$.
  \begin{align*}
    \sem{\Gamma, x : \typeA, y : \typeB, \Delta \vljud t : \typeC} & =
    \sem{\Gamma, y : \typeB, x : \typeA, \Delta \vljud t : \typeC}
    \comp \exch_{\Gamma, \underline{x : \typeA, y : \typeB}, \Delta} \\
    \sem{\Gamma, \Delta \vljud s[w / x] : \typeB} & =
    \sem{\Gamma, x : \typeA \vljud s : \typeB} \comp
    \join_{\Gamma;\typeA} \comp\, (\id \otimes \sem{\Delta \vljud w : \typeA})
    \comp \spl_{\Gamma;\Delta}
  \end{align*}
\end{lemma}

\begin{theorem}
  \label{theo:bsound}
  The equations presented in Figure~\ref{fig:eqs} are sound w.r.t.\  judgement
  interpretation (Figure~\ref{fig:typing_rules_cond}). More specifically, if
  $\Gamma \vljud t = s : \typeA$ is one of the equations in
  Figure~\ref{fig:eqs} then $\sem{\Gamma \vljud t : \typeA} = \sem{\Gamma
  \vljud s : \typeA}$.
\end{theorem}

\subsection{Interpretation of $\Quant$-(in)equations}
Let us now move to  the intepretation of $\Quant$-(in)equations. Essentially,
the main idea is to move from \emph{ordinary} co-Cartesian autonomous
categories to suitably enriched ones. We thus start by describing a certain
basis of enrichment, that of `generalised metric
spaces'~\cite{hofmann14,stubbe14}.
\begin{definition}
  \label{defn:vcat}
  A $\Quant$-category is a pair $(X,a)$ where $X$ is a set and $a : X \times X
  \to \Quant$ is a function that satisfies
  \begin{flalign*}
    k \leq a(x,x) \qquad \text{ and }  \qquad
    a(x,y) \otimes a(y,z) \leq a(x,z) \hspace{2cm}
    (x,y,z \in X)
  \end{flalign*}
  For two $\Quant$-categories $(X,a)$ and $(Y,b)$, a $\Quant$-functor $f :
  (X,a) \to (Y,b)$ is a function $f : X \to Y$ that satisfies 
  $a(x,y) \leq b(f(x),f(y))$ for all $x,y \in X$.
\end{definition}
$\Quant$-categories and $\Quant$-functors form a category which we denote by
$\VCat$.  A $\Quant$-category $(X,a)$ is called \emph{symmetric} if $a(x,y) =
a(y,x)$ for all $x,y \in X$. We denote by $\VCatSy$ the full subcategory of
$\VCat$ whose objects are symmetric. Also, every $\Quant$-category carries a
natural order defined by $x \leq y$ whenever $k \leq a(x,y)$. A
$\Quant$-category is called \emph{separated} if its natural order is
anti-symmetric. We denote by $\VCatSe$ the full subcategory of $\VCat$ whose
objects are separated, and by $\VCatSS$ the full subcategory of $\VCat$ whose
objects are both symmetric and separated. Note as well the following reflections,
\[
        \VCatSe \adjuncte{\mathsf{sep}}{} \VCat 
        \hspace{1cm} 
        \VCatSS
        \adjuncte{\mathsf{sep}}{} 
        \VCatSy
\]
Specifically the left adjoint the $\mathsf{sep} : \VCat \to \VCatSe$ (and
respective restriction to symmetric categories) is built in two
stages~\cite{hofmann14}. Given a $\Quant$-category $(X,a)$, one defines the
equivalence relation $x \sim y$ whenever $x \leq y$ and $y \leq x$ (where
$\leq$ is the natural order introduced earlier). This relation induces the
\emph{separated} $\Quant$-category $(X/_\sim, \tilde a)$, where $\tilde a$ is
defined by $\tilde a([x],[y]) = a(x,y)$ for all $[x],[y] \in X/_\sim$. 
\begin{example}
  \label{ex:quant}
  For $\Quant$ the Boolean quantale, $\VCat$ is the category of preordered sets
  and monotone maps. $\VCatSe$ is the full subcategory of partially ordered
  sets and $\VCatSS$ is the category of sets and functions.  For
  $\Quant$ the metric quantale, $\VCat$ is the category of hemimetric spaces
  and non-expansive maps~\cite{fernandez23}. $\VCatSy$ is the full subcategory
  of pseudometric spaces, $\VCatSe$ the full subcategory of quasimetric
  spaces~\cite{cobzas12}, and finally $\VCatSS$ is the full subcategory $\Met$ of
  metric spaces. An analogous reasoning applies to the ultrametric quantale.  
\end{example}
The categories $\VCat$, $\VCatSy$, $\VCatSe$, and $\VCatSS$ are
autonomous~\cite{hofmann14,day72}, and therefore are a suitable basis of
enrichment~\cite{kelly82}.  If the underlying quantale is integral (\ie\ if the
top element is $k$), the autonomous structure of all these categories is
particularly easy to describe.  The tensor is defined as $(X,a) \otimes (Y,b)
:= (X \times Y, a \otimes b)$ where $a\otimes b$ is given by,
\begin{flalign*}
  (a \otimes b)((x,y), (x',y')) = a(x,x') \otimes b(y,y')
\end{flalign*}
and the set of $\Quant$-functors between $(X,a)$ and $(Y,b)$ comes equipped
with the map,
\[
(f,g) \mapsto \bigwedge_{x \in X} b(f(x),g(x))
\]
\begin{example}
All examples in Example~\ref{ex:quant} are based on an integral
quantale.
\end{example} 
Note that for any quantale $\Quant$ these four categories also have binary
products, defined as $(X,a) \times (Y,b) := (X \times Y, a \wedge b)$ where,
\begin{flalign*}
  (a \wedge b)((x,y), (x',y')) = a(x,x') \wedge b(y,y')
\end{flalign*}
They also have binary coproducts, defined as $(X,a) + (Y,b) := (X + Y, c)$
where $c$ is the smallest function that makes the injections $\inl$ and $\inr$
$\VCat$-functors.  Unlike in~\cite{dahlqvist22,dahlqvist23} this (co-)Cartesian
structure plays a key rôle in the present paper, as detailed later on. 

\begin{definition}\label{def:VCatEnriched}
  A category $\catC$ is $\VCat$-enriched (or simply, a $\VCat$-category) if for
  all $\catC$-objects $X$ and $Y$ the hom-set $\catC(X,Y)$ is a
  \emph{$\Quant$-category} and if the composition of $\catC$-morphisms is a
  $\Quant$-functor,
  \begin{flalign*}
    (\ \cdot\ ) : \catC(X,Y) \otimes \catC(Y,Z)
    \longrightarrow \catC(X,Z)
  \end{flalign*}
  Given two $\VCat$-categories $\catC$ and $\catD$ and a functor $\funF : \catC
  \to \catD$, we call $\funF$ a $\VCat$-enriched functor (or simply,
  $\VCat$-functor) if for all $\catC$-objects $X$ and $Y$ the map $\funF_{X,Y}
  : \catC(X,Y) \to \catD(\funF X, \funF, Y)$ is a $\Quant$-functor.  An
  adjunction $\catC : \funF \dashv \funG : \catD$ is called $\VCat$-enriched if
  for all objects $X \in \catC$ and $Y \in \catD$ there exists a
  $\Quant$-isomorphism $\catD(\funF X, Y) \cong \catC(X, \funG Y)$ natural in
  $X$ and $Y$.  We obtain analogous notions of enrichment by substituting
  $\VCat$ with $\VCatSe$, $\VCatSy$, or $\VCatSS$.
\end{definition}

\begin{definition}\label{defn:enr_aut}
  A $\VCat$-enriched autonomous category $\catC$ is an autonomous and
  $\VCat$-enriched category $\catC$ such that the bifunctor $\otimes : \catC
  \times \catC \to \catC$ is a $\VCat$-functor and the adjunction $(- \otimes
  X) \dashv (X \multimap -)$ is a $\VCat$-adjunction. Again we obtain analogous
  notions of an enriched autonomous category by replacing $\VCat$ (as basis of
  enrichment) with $\VCatSe$, $\VCatSy$, or $\VCatSS$.
\end{definition}
The previous definition presents an enrichment of a category's autonomous
structure -- and, in the new setting of additive disjunction, we ought to  
work with an enrichment of coproducts as well. However instead of enriching the
latter over the $\otimes$ structure of $\VCat$ we enrich over the Cartesian
structure ($\times$) instead. This may look \emph{ad-hoc}, but from a logical
perspective it is the natural approach: while the multiplicative (\ie\
$\otimes$) structure of $\lambda$-calculus is handled via an enrichment over
the $\otimes$ structure of $\VCat$, the additive (\ie\ Cartesian) structure is
handled via an enrichment over the Cartesian structure of $\VCat$. More
formally our enrichment of coproducts amounts to imposing that the co-pairing
map is a $\Quant$-functor,
\begin{equation}
        \label{eq:coprod}
        [-,=] : \catC(X,Y) \times \catC(Z,Y) \to \catC(X + Z, Y)
\end{equation}
\begin{definition}
        \label{defn:auto_coprod}
        A $\VCat$-co-Cartesian autonomous category is a $\VCat$-autonomous
        category with binary coproducts enriched over the product structure of
        $\VCat$ (\ie\ we have the $\Quant$-functor~\eqref{eq:coprod}). Again we
        obtain analogous notions of enriched co-Cartesian autonomous category
        by replacing $\VCat$ (as basis of enrichment) with $\VCatSe$,
        $\VCatSy$, or $\VCatSS$.
\end{definition}

Recall the interpretation of $\lambda$-terms (from
Figure~\ref{fig:typing_rules_cond}) and let $\catC$ be a $\VCat$-co-Cartesian
autonomous category. We say that a $\Quant$-inequation $\Gamma \vljud t \leq_q
s : \typeA$ is \emph{satisfied} by this interpretation if $q \leq a(\sem{\Gamma
\vljud t : \typeA},\sem{\Gamma \vljud s : \typeA})$ where $a :
\catC(\sem{\Gamma},\sem{\typeA}) \times \catC(\sem{\Gamma},\sem{\typeA}) \to
\Quant$ is the underlying function of the $\Quant$-category
$\catC(\sem{\Gamma},\sem{\typeA})$. This notion of satisfaction and the
previous definition also give way to the following notions of theory and model.

\begin{definition}[$\Quant\lambda$-theory]\label{defn:theory}
  Consider a tuple $(G,\Sigma)$ consisting of a set $G$ of ground types and a
  set $\Sigma$ of sorted operation symbols.  A \emph{$\Quant \lambda$-theory}
  $((G,\Sigma),Ax)$ is a triple such that $Ax$ is a set of
  $\Quant$-inequations-in-context between $\lambda$-terms built from
  $(G,\Sigma)$.
\end{definition}

As usual, the elements of $Ax$ are called the \emph{axioms} of the theory. In
this setting, a theorem is a $\Quant$-inequation that is provable from the laws
listed in Figure~\ref{fig:eqs} and Figure~\ref{fig:equations-in-context-cond}.
If the aforementioned symmetry law is imposed then we speak of a symmetric
$\Quant\lambda$-theory.
  
\begin{definition}[Models of (symmetric) $\Quant \lambda$-theories]\label{defn:model}
        Let us consider a theory $\mathscr{T}$ and a $\VCatSe$-co-Cartesian
        autonomous category $\catC$. Suppose that for each $X \in G$ we have an
        interpretation $\sem{X}$ that is a $\catC$-object and analogously for
        the operation symbols. This interpretation structure is a \emph{model}
        of the theory if all axioms are satisfied by the interpretation. Models
        of symmetric theories follow the same reasoning except that one
        replaces the basis of enrichment $\VCatSe$ by $\VCatSS$.
\end{definition}

\begin{remark}
        While we do not pursue such a quest here, note that the imposition of
        separability in the previous definition is crucial for making the
        semantics of our extended calculus functorial (\ie\ for framing our
        semantics as a certain kind of functor instead of directly using the
        rules in Figure~\ref{fig:typing_rules_cond}). In fact, references
        \cite{dahlqvist22,dahlqvist23} show that the original calculus already
        requires separability.
\end{remark}

\subsection{Soundness and completeness}

Let us now establish soundness and approximate completeness.

\begin{theorem}[Soundness]\label{theo:sound}
  Consider a (symmetric) $\Quant\lambda$-theory $\mathscr{T}$ and a
  corresponding model. If we have $\Gamma \vljud t \leq_q s : \typeA$ as a
  theorem of $\mathscr{T}$ then $q \leq a(\sem{\Gamma \vljud t :
  \typeA},\sem{\Gamma \vljud s : \typeA})$.
\end{theorem}

\begin{proof}
  Follows straightforwardly by induction over proof derivation trees.
\end{proof}
The case of approximate completeness is less straightforward: among other
things, due to the presence of additive disjunction, we will need that for
every $q \in \Quant$ the operation,
\begin{align}
        \label{op:dist}
        q \wedge (-) : \Quant \to \Quant
\end{align}
preserves directed suprema. Fortunately such is indeed the case when $\Quant$
is continuous~\cite[Proposition I-1.8]{gierz03}.  In detail,

\begin{theorem}[Approximate completeness]
        Consider a (symmetric) $\Quant \lambda$-theory $\mathscr{T}$. If a given
        inequation $\Gamma \vljud t \leq_q s : \typeA$ holds in all models of
        $\mathscr{T}$ then for all \emph{approximations} $r \ll q$ ($r \in B$) we have
        $\Gamma \vljud t \leq_r s : \typeA$ as a theorem of $\mathscr{T}$. In
        particular, if $q$ is compact (\ie\ $q \ll q$) we have $\Gamma \vljud t
        \leq_q s : \typeA$ as a theorem of $\mathscr{T}$.
\end{theorem}

\begin{proof}[Proof]
  The proof is based on the idea of a Lindenbaum-Tarski algebra.  Specifically
  it follows from,
  \begin{enumerate}
          \item building a $\VCatSe$-enriched `syntactic' category $\Syn(\mathscr{T})$;
          \item showing that it is a model of $\mathscr{T}$;
          \item and showing that if $q \leq a(\sem{t},\sem{s})$ in
                  $\Syn(\Theo)$ then for all approximations $r \ll q$ $(r \in B)$ we have
                  $\Gamma \vljud t \leq_r s : \typeA$ as a theorem of $\mathscr{T}$.
  \end{enumerate}
  We overview each step. Step (1) was actually already described
  in~\cite{dahlqvist22,dahlqvist23}.  The main idea is as follows: given two
  types $\typeA$ and $\typeB$ we take set of judgments of the form $x
  : \typeA \vljud t : \typeB$ and equip it with the structure of a
  $\Quant$-category by setting,
  \begin{equation}
          \label{eq:compl}
        a \big (x: \typeA \vljud t : \typeB ,y : \typeA \vljud s : \typeB \big ) 
        = \bigvee \big \{ q \mid t[z/x] \leq_q s[z/y] \text{ a theorem of } \Theo \big \}
  \end{equation}
  We then quotient this structure into a \emph{separated} $\Quant$-category via
  the previous left adjoint $\VCat \to \VCatSe$ (or its restriction $\VCatSy
  \to \VCatSS$ if $\mathscr{T}$ is symmetric). Such a process forms the
  hom-objects of $\Syn(\Theo)$, whose objects are the types in $\mathscr{T}$
  and composition is as usual substitution.  Note that the set from which we
  take the supremum \eqref{eq:compl} is directed, thanks to
  Rule~\rulename{join} from Figure~\ref{fig:equations-in-context-cond}.

  Concerning Step (2), \cite{dahlqvist22,dahlqvist23} already show that
  $\Syn(\Theo)$ is $\VCatSe$-autonomous, so what is left to prove is that
  $\Syn(\Theo)$ is $\VCatSe$-co-Cartesian as well (recall
  Definition~\ref{defn:model}). Thus consider two terms $a : \typeA \vljud t :
  \typeC$ and $b : \typeB \vljud s : \typeC$. We define their co-pairing  by,
  \[
          z : \typeA \oplus \typeB \vljud
          \ccelim{z}{t[x/a]}{s[y/b]} : \typeC
  \]
  The fact that this co-pairing construction is congruent follows from meets
  being distributive over directed joins (recall Operation~\eqref{op:dist}).
  In fact, the same property also entails that the resulting coproducts are
  enriched over the Cartesian structure of $\VCatSe$ (resp. $\VCatSS$ if $\mathscr{T}$
  is symmetric). Let us prove this last case in detail (the former one
  is analogous so we do not detail it here). We need to prove that,
  \[
          a(t,t') \wedge a(s,s') \leq a([t,s], [t',s'])
  \]
  for all $\lambda$-terms $a : \typeA \vljud t : \typeC$, $a' : \typeA \vljud
  t' : \typeC$, $b : \typeB \vljud s : \typeC$, $b' : \typeB \vljud s' :
  \typeC$.  By virtue of the previous distributivity property and the
  $\Quant$-congruence rules (from Figure~\ref{fig:equations-in-context-cond}),
  we obtain the sequence of (in)equalities below, which proves our claim.
  \begin{align*}
          a(t,t') \, \wedge \, a(s,s') & = 
          \bigvee \big \{ q \mid t[z/a] \leq_q t'[z/a'] 
                  \text{ a theorem of } \mathscr{T}
                  \big \}
          \, \wedge \,
          \bigvee \big \{ r \mid s[w/b] \leq_r s'[w/b'] \text{ a theorem of } \mathscr{T}
                \big \}
          \\
          &
          = \bigvee \big \{ q \wedge r \mid t[z/a] \leq_q t'[z/a'] 
                  \text{ a theorem of } \mathscr{T} \text{ and }
                  s[w/b] \leq_r s'[w/b'] \text{ a theorem of } \mathscr{T}
                \big \}
          \\
          &
          \leq \bigvee \big \{
                  q \mid [t,s] \leq_q [t',s'] \text{ a theorem of } \mathscr{T}
          \big \}
          \\
          &
          = a([t,s],[t',s'])
  \end{align*}
  Finally Step (3) boils down to showing that if an inequation $x : \typeA
  \vljud t \leq_q s : \typeB$ is satisfied by $\Syn(\mathscr{T})$ then for all
  approximations $r \ll q$ $(r \in B)$ the inequation $x : \typeA  \vljud t
  \leq_r s : \typeB$ is a theorem of $\Theo$. So by assumption $q \leq
  a([t],[s]) = a(t,s) = \bigvee \{ p \mid t \leq_p s \text{ a theorem of }
  \mathscr{T}\}$. 
  It follows from the definition of the way-below relation that there exists a
  \emph{finite} set $F \subseteq \{ p \mid t \leq_p s \text{ a theorem of }
  \mathscr{T} \}$ such that $r \leq \bigvee F$. Thus by an application of
  Rule~\rulename{join} in Figure~\ref{fig:equations-in-context-cond} we obtain
  $t \leq_{\bigvee F} s$ and consequently Rule~\rulename{weak} yields $t \leq_r
  s$ as a theorem of $\Theo$.
\end{proof}

\section{A fistful of models: from gluing to enriched presheaves}
\label{sec:models}
\subsection{Generalised metric spaces} 
We now present a series of models of our calculus (recall
Definition~\ref{defn:auto_coprod}), starting with the category $\VCat$ of
`generalised metric spaces' and the aforementioned full subcategories. First,
the category $\VCat$ is $\VCat$-autonomous and also (co)complete (see
\eg~\cite{hofmann14,dahlqvist22,dahlqvist23} and the previous section). Its
(co)limits are in fact $\VCat$-enriched, in the sense that there exist
$\VCat$-isomorphisms,
\begin{align*}
        \label{eq:enriched}
        \VCat(X, \lim \Diag) \cong \lim \VCat(X, \Diag) 
        \hspace{1cm}
        \VCat(\colim \Diag, X) \cong \lim \VCat(\Diag, X)
\end{align*}
and not just bijections, as in ordinary category theory (note that the second
case, in particular, subsumes our requirement of binary coproducts being
enriched over the Cartesian structure of $\VCat$). The proof of these two cases
is as follows: while the case of limits follows directly from exponentials
preserving limits, the case of colimits is entailed by the following Yoneda
argument in the category of presheaves $[\VCat^\cop,\Set]$,
\begin{align*}
        \VCat(Y, \VCat(\colim \Diag, X)) 
        &
        \cong \VCat(Y \otimes \colim \Diag, X)
        \\
        &
        \cong \VCat(\colim Y \otimes \Diag, X)
        \\
        &
        \cong \lim \VCat(Y \otimes \Diag, X)
        \\
        &
        \cong \lim \VCat(Y, \VCat(\Diag, X))
        \\
        &
        \cong \VCat(Y, \lim \VCat(\Diag, X))
\end{align*}
with the whole sequence of isomorphisms natural in $Y$. Let us now transport
all this rich structure to the full subcategories $\VCatSe$, $\VCatSy$, and
$\VCatSS$. Our method will hinge on the following adjoint functors,
\[
        \VCatSe \adjuncte{\mathsf{sep}}{} \VCat \adjuncteb{}{\mathsf{sym}} \VCatSy
        \adjunctopebb{}{\mathsf{sep}} \VCatSS
\]
Recall that the left adjoint $\mathsf{sep} : \VCat \to \VCatSe$ (resp. its
restriction to symmetric categories) is built in two stages~\cite{hofmann14}.
Given a $\Quant$-category $(X,a)$, one defines the equivalence relation $x \sim
y$ whenever $x \leq y$ and $y \leq x$ (where $\leq$ is the natural order
introduced earlier). This relation induces the \emph{separated}
$\Quant$-category $(X/_\sim, \tilde a)$, where $\tilde a$ is defined by $\tilde
a([x],[y]) = a(x,y)$ for all $[x],[y] \in X/_\sim$. Then, by general
results~\cite{cats} the existence of a reflection entails that $\VCatSe$ is
closed under the $\VCat$-enriched limits of $\VCat$. Not only this, the
definition of $\mathsf{sep}$ discloses almost immediately that this functor is
$\VCat$-enriched. Thus $\VCatSe$ inherits $\VCat$-enriched colimits from
$\VCat$ as well. As the next step in our quest, observe that if a
$\Quant$-category $Y$ is separated then $\VCat(X,Y)$ is separated as well for
any $\Quant$-category $X$.  One can thus appeal to Day's reflection
theorem~\cite{day72,malherbe13,lack14} to establish a $\VCat$-enriched
autonomous structure for $\VCatSe$. Putting everything together, we conclude
that $\VCatSe$ is a $\VCat$-co-Cartesian autonomous category, and in particular
a $\VCatSe$-co-Cartesian autonomous category.

The functor $\mathsf{sym} : \VCat \to \VCatSy$ sends a $\Quant$-category
$(X,a)$ into $(X, a \wedge a \comp \sw)$. While $\mathsf{sym}$ is not
$\VCat$-enriched, and thus one cannot proceed as categorically as before,
limits in $\VCat$ have an amenable description that makes it easy to see that
they are closed under symmetry. This means that $\VCatSy$ is closed under
$\VCat$-enriched limits from $\VCat$. On the other hand, dually to before, the
existence of a coreflection (\emph{viz.} $\mathsf{sym}$) entails that $\VCatSy$
is closed under the $\VCat$-enriched colimits of $\VCat$~\cite{cats}. Next, the
tensor of $\VCat$ is closed under symmetry and if a $\Quant$-category $Y$ is
symmetric then $\VCat(X,Y)$ is symmetric as well for any $\Quant$-category $X$.
Thus $\VCatSy$ inherits the $\VCat$-autonomous structure of $\VCat$. Putting
again everything together, we deduce that $\VCatSy$ is a $\VCat$-co-Cartesian
autonomous category, and in particular a $\VCatSy$-co-Cartesian autonomous
category. It remains to tackle the case $\VCatSS$, but the underlying reasoning
is essentially the same as the one we applied to $\VCatSe$.

All together, the reasoning of this subsection yields a large collection of
models our extended calculus (see Example~\ref{ex:quant} and
Example~\ref{ex:quant2}).
\subsection{Gluing construction}
        The Artin gluing construction is a well-known categorical technique for
        proving properties of formal systems, such as logics, type theories,
        and programming languages~\cite{crole93,hyland03}.  It relies on
        \emph{comma categories}, more specifically on $\Id \downarrow G$ for a
        given functor $G : \catC \to \catD$~\cite{maclane71}. It is usually
        required that $\Id \downarrow G$ is a model of the system at hand.
        Thus in the case of linear $\lambda$-calculus with additive
        disjunction, it should be co-Cartesian autonomous.
        Reference~\cite{hyland03} shows that such is indeed the case if $\catC$
        and $\catD$ are co-Cartesian autonomous, $G$ is lax monoidal, and
        $\catD$ has pullbacks. Our quantalic calculus, on the other hand,
        requires an extension of this result to the $\VCat$-enriched setting --
        this is what we will prove next.

        First recall that objects in $\Id \downarrow G$ are triples
        $(X,\phi,Y)$, with $X \in \catD$, $Y \in \catC$, and $\phi : X \to GY$
        a $\catD$-morphism. Also, morphisms between $(X,\phi,Y)$ and $(W,\psi,Z)$ are
        pairs $(f : X \to W, g : Y \to Z)$ such that $\psi \comp f = Gg \comp
        \phi$. Now, let $\catC$ and $\catD$ be $\VCat$-co-Cartesian autonomous
        categories, $G$ a lax monoidal $\VCat$-functor, and $\catD$ possess
        $\VCat$-enriched pullbacks. The core observation is that $\Id
        \downarrow G$ can be naturally enriched by taking in $\VCat$ the
        pullbacks below as hom-objects,
        \[
        \begin{tikzcd}[column sep=30pt, row sep=30pt]
                \Id \downarrow G\bigl((X,\phi,Y),(W,\psi,Z)\bigr)
                \arrow[r]
                \arrow[d]
                \arrow[dr, phantom, "\lrcorner" very near start]
                &
                \catC(Y,Z)
                \arrow[d, "((-)\comp\phi)\comp G_{Y,Z}"]
                \\
                \catD(X,W)
                \arrow[r, "\psi \comp (-)"']
                &
                \catD(X,GZ)
        \end{tikzcd}
        \]
        The remaining steps for showing that $\Id \downarrow G$ is an instance
        of Definition~\ref{defn:auto_coprod} boil down to `internalising'  the
        steps taken in~\cite{hyland03} in our basis of enrichment (\ie\
        $\VCat$). For example, if the tensor $\Big ( \Id \downarrow \Gamma(
        (X,\phi,Y) , (X',\psi,Y') ) \Big ) \otimes \Big ( \Id \downarrow
        \Gamma( (X',\psi,Y') , (X'',\chi,Y'') ) \Big )$ is abbreviated to $P$,
        composition is described by the following diagram, 
        \[
        \begin{tikzcd}[column sep=30pt, row sep=30pt]
                P
                \arrow[rr, "\pi_2 \otimes \pi_2"]
                \arrow[dr, dashed]
                \arrow[dd, "\pi_1 \otimes \pi_1"']
                & &
                \catC(Y,Y') \otimes \catC(Y',Y'')
                \arrow[d, "(\,.\,)"]
                \\
                &
                \Id \downarrow \Gamma\bigl((X,\phi,Y),(X'',\chi,Y'')\bigr)
                \arrow[r]
                \arrow[d]
                \arrow[dr, phantom, "\lrcorner" very near start]
                &
                \catC(Y,Y'')
                \arrow[d, "((-)\comp\phi)\comp G_{Y,Y''}"]
                \\
                \catD(X,X') \otimes \catD(X',X'')
                \arrow[r, "(\,.\,)"']
                &
                \catD(X,X'')
                \arrow[r, "\chi \comp (-)"']
                &
                \catD(X,GY'')
       \end{tikzcd}
       \]
       The case of currying is arguably the most complicated one to
       internalise, so let us sketch the construction.  First, the tensor of
       two objects $(X,\phi,Y)$ and $(X',\psi,Y')$ is given by $(X \otimes X'
       , m \comp (\phi \otimes \psi), Y \otimes Y')$ where $m_{A,B} : GA
       \otimes GB \to G(A \otimes B)$ is the natural transformation associated
       to $G$ being lax monoidal.  Second, the exponential of $(X,\phi,Y)$ and
       $(X',\psi,Y')$ is given by the top arrow in the ($\VCat$-enriched)
       pullback of $\catD$,
       \[
        \begin{tikzcd}[column sep=40pt, row sep=40pt]
                (X \multimap X') \times_p G(Y \multimap Y')
                \arrow[r]
                \arrow[d]
                \arrow[dr, phantom, "\lrcorner" very near start]
                &
                G(Y \multimap Y')
                \arrow[d, "((-) \comp \phi) \comp \overline{(G\app \comp m)}"]
                \\
                X \multimap X'
                \arrow[r, "\psi \comp (-)"']
                &
                X \multimap GY'
        \end{tikzcd}
        \]
        Now, observe that for any $Z \in \catD$ the functor $\catD(Z,-) : \catD
        \to \VCat$ must preserve this limit (a pullback), since the latter is
        $\VCat$-enriched. Finally consider $(X,\phi,Y)$, $(X',\psi,Y')$,
        $(X'',\chi,Y'')$ and let $P$ be hom-object of morphisms $(X, \phi, Y)
        \otimes (X',\psi,Y') \to (X'',\chi,Y'')$.  These data induce the
        diagram of pullbacks below, where the pullback in the bottom right
        corner is given precisely by the fact that $\catD(Z,-)$ is continuous
        w.r.t. pullbacks.
        \[
        \begin{tikzcd}[column sep=30pt, row sep=30pt]
                P
                \arrow[rr]
                \arrow[dr, dashed]
                \arrow[ddr, dashed, bend right=15]
                \arrow[d]
                & &
                \catC(Y \otimes Y',Y'')
                \arrow[d, "\overline{(-)}"]
                \\
                \catD(X \otimes X', X'')
                \arrow[dd, "\overline{(-)}"']
                & 
                \Id \downarrow G((X,\phi,Y),(X',\psi,Y') \multimap (X'',\chi,Y''))
                \arrow[d]
                \arrow[r]
                \arrow[dr, phantom, "\lrcorner" very near start]
                &
                \catC(Y, Y' \multimap Y'')
                \arrow[d, "((-) \comp \phi) \comp G_{Y,Y' \multimap Y''}"]
                \\
                &
                \catD \bigl(X,(X' \multimap X'') \times_p G(Y' \multimap Y'') \bigr)
                \arrow[r, "\pi_2 \comp (-) "]
                \arrow[d, "\pi_1 \comp (-) "']
                \arrow[dr, phantom, "\lrcorner" very near start]
                &
                \catD(X,G(Y' \multimap Y''))
                \arrow[d, "((-)\comp\psi)\comp \overline{G \app \comp m} \comp (-)"]
                \\
                \catD(X,X' \multimap X'')
                \arrow[r, equal]
                &
                \catD(X,X' \multimap X'')
                \arrow[r, "(\chi \comp (-)) \comp (-)"']
                &
                \catD(X,X' \multimap GY'')
       \end{tikzcd}
       \]
       The fact that the outer diagram commutes means that for every pair of
       morphisms $(f,g) \in P$ the following equation holds,
       \[
           (- \comp \psi) \comp \overline{G\app \comp m} \comp
           G\, \overline{g} \comp \phi = (\chi \comp -) \comp \overline{ f }
      \]
      We prove this via the following reasoning, which uses the universal
      property of exponentials and the monoidality of $G$,
        \begin{align*}
        (- \comp \psi) \comp \overline{G \app \comp m}  
        \comp G\,\overline{g} \comp \phi
        & =
        \overline{ G \app \comp m \comp (\id \otimes \psi)} 
        \comp G \,\overline{g} \comp \phi 
        \\
        & =
        \overline{ 
                G \app \comp m \comp (\id \otimes \psi) 
                \comp ((G \, \overline{g} \comp \phi) \otimes \id ) 
        }
        \\
        & = 
        \overline{
                G\app \comp m
                \comp (G \, \overline{g} \otimes \id) \comp (\phi \otimes \psi)
        }
        \\
        & =
        \overline{
                G \app 
               \comp G \,(\overline{g} \otimes \id) \comp m \comp (\phi \otimes \psi)
        }
        \\
        & =
        \overline{ Gg
        \comp m \comp (\phi \otimes \psi)}
        \\
        & = 
        \overline {\chi \comp f}
        \\
        & = 
        (\chi \comp -) \comp \overline{f}
        \end{align*}
        This establishes the existence of the lower mediating morphism which
        maps $(f,g)$ into $\pv{\overline{f}}{G\overline{g} \comp \phi}$. The
        upper mediating morphism then emerges directly.

\subsection{Enriched presheaves}
        \label{ex:presheaves}
Let us start by detailing the $\VCat$-category $[\catC^\cop, \VCat$] of enriched presheaves
where $\catC$ is a small, $\VCat$-symmetric monoidal category. Morphisms are
given by the $\VCat$-enriched limit~\cite{borceux94},
\[
        [\catC^\cop,\VCat](P,Q) \cong \int_X \VCat(P(X),Q(X))
\]
The category has $\VCat$-enriched (co)limits which are inherited from $\VCat$ via
pointwise extension. Not only this, it is $\VCat$-autonomous with the tensor
and exponential given by the $\VCat$-enriched (co)limits~\cite{day72},
\[
        P \otimes Q \cong \int^{X,Y} P(X) \otimes Q(Y) \otimes 
        \catC(-, X \otimes Y)
        \hspace{1.5cm}
        P \multimap Q \cong \int_X \VCat(P(X), Q(- \otimes X))
\]
This provides a series of examples of Definition~\ref{defn:auto_coprod}, and as
before, analogous examples arise by replacing the basis of enrichment $\VCat$
with any of the aforementioned full subcategories. 

Suppose now that $\catC$ is additionally $\VCat$-co-Cartesian, \ie\ that it has
binary coproducts enriched over the Cartesian structure of $\VCat$, and that
the tensor $\otimes$ distributes over them. Ideally  Yoneda
embedding $\catC \to [\catC^\cop,\VCat]$ would these preserve coproducts, which
unfortunately is false. In the \emph{non-enriched} setting, a way around this
problem is to work instead with the full subcategory of \emph{finite-product
preserving} presheaves and the respective restriction of Yoneda
embedding~\cite{fu22,malherbe13}. This subcategory is reflective. Thus it is
autonomous and has binary coproducts (inherited by the supercategory), as
needed for our calculus~\cite{day72}. We will show that this situation lifts to
the enriched setting.  Specifically we will show the existence of an analogous
reflection,
\[
        [\catC^\cop,\VCatSe]_{\fp} \adjuncte{\mathscr{L}}{i} [\catC^\cop,\VCatSe] 
\]
in the $\VCatSe$-enriched setting (and likewise for $\VCatSS$). This provides
yet another set of examples of Definition~\ref{defn:auto_coprod}, with the
restriction of Yoneda embedding $\catC \to [\catC^\cop,\VCatSe]_\fp$ now
preserving coproducts.

Interestingly the description of this left adjoint relies on our
$\Quant$-equational system adjusted to the setting of \emph{multi-sorted
universal algebra}~\cite{rosicky94}. We start by describing its action on
objects. Thus take a presheaf $P$ in $[\catC^\cop,\VCatSe]$. It induces the
following algebraic signature (in the sense of~\cite{rosicky94}).
\begin{itemize}
        \item the types are the objects of $\catC^\cop$;
        \item for every morphism $f : X \to Y$ in $\catC^\cop$ we have
                an operation symbol $f : X \to Y$;
        \item for every $c \in P(X)$ we have a constant $c : X$;
        \item for every product cone $\pi_i : X_1 \times X_2 \to X_i$ in
                $\catC^\cop$ we have an operation $(-,=) : X_1,X_2 \to X_1
                \times X_2$.
\end{itemize}
Note that in this restricted setting of universal algebra, ground types are the
only types available (\ie\ we do not form new types from simpler ones). Note as
well that the judgement formation rules are now non-linear and boil down to,
\[
      \begin{prooftree}
              \hypo{x : \typeA \in \Gamma}
              \infer1[\rulename{hyp}]{ \Gamma  \vljud x : \typeA}
      \end{prooftree}
      \hspace{2.5cm}
      \begin{prooftree}
                    \hypo{\Gamma \vljud t_i : \typeA}
                    \hypo{f \in \typeA_1, \dots , \typeA_n \to \typeA \in \Sigma}
                    \infer2[\rulename{ax}]{\Gamma \vljud f(t_1,\dots,t_n) : \typeA}
      \end{prooftree}
\]
Now, for every $\catC$-object $X$ the carrier of $\mathscr{L}(P)(X)$ is the set
of closed terms $- \vljud t : X$ modulo an equivalence relation based on a
certain $\Quant$-axiomatic schema. Specifically the latter is given by the
rules above the dotted line in Figure~\ref{fig:theo_rules}. The rules below the
dotted line correspond to the $\Quant$-(in)equational system of multi-sorted
universal algebra.  As usual we call \emph{theorems} those (in)equations
$\Gamma \vljud t \leq_q s : \typeA$ that can be deduced from the rules in
Figure~\ref{fig:theo_rules}. And finally for every $\catC$-object $X$,
$\mathscr{L}(P)(X)$ is the \emph{separated} $\Quant$-category defined by,
\[
        \Big (\{ t \mid - \vljud t : X \}, 
        (t,s) \mapsto \bigvee \{ q \mid t \leq_q s \text{ is a theorem} \} \Big )_\mathsf{sep}
\]
This yields an enriched, finite-product preserving functor by construction and
a natural transformation $\eta_P : P \to i \mathscr{L}(P)$ as well.
\begin{figure}[h!]
\[
      \begin{prooftree}
              \hypo{c \in X, d \in Y}
              \hypo{f : X \to Y \in \catC^\cop}
              \hypo{P(f)(c) = d}
              \infer3[]{- \vljud f(c) = d  : Y}
      \end{prooftree}
      \hspace{1cm}
      \begin{prooftree}
              \hypo{\Gamma \vljud t : X}
              \infer1[]{\Gamma \vljud \id(t) = t : X}
      \end{prooftree}
      \hspace{1cm}
      \begin{prooftree}
              \hypo{\Gamma \vljud t : X}
              \infer1[]{\Gamma \vljud g \comp f(t) = g(f(t)) : X}
      \end{prooftree}
\]

\vspace{0.5cm}
\noindent
For every $\catC^\cop$-product cone $\pi_i : X_1 \times X_2 \to X_i$
\[
        \begin{prooftree}
                \hypo{\Gamma \vljud t_i : X_i}
                \infer1[]{\Gamma \vljud \pi_i \, (t_1,t_2) = t_i : X_i}
        \end{prooftree}
        \hspace{2.5cm}
        \begin{prooftree}
                          \hypo{\Gamma \vljud t : X_1 \times X_2}
                          \infer1[]{ \Gamma \vljud ( \pi_1 \, t , \pi_2 \, t ) = t : X_1 \times X_2}{}
        \end{prooftree}
\]

\vspace{0.5cm}
For all $\catC^\cop$-morphisms $f,g : X \to Y$, elements $c,d \in P(X)$, and elements $q \in B$
\[
        \begin{prooftree}
                \hypo{q \leq a(c,d)}
                \infer1[]{- \vljud c \leq_q d : X}
        \end{prooftree}
        \hspace{2.5cm}
        \begin{prooftree}
                          \hypo{q \leq a(f,g)}
                          \infer1[]{ x : X \vljud f(x) \leq_q g(x) : Y}
        \end{prooftree}
\]

    \vspace{0.3cm}
    \dotfill

\[
        \infer[(\textbf{refl})]{t \leq_k t}{}
        \qquad \qquad
        \infer[(\textbf{trans})]{t \leq_{q \otimes r} u}{ t \leq_q s \quad s \leq_r u}
        \qquad \qquad
        \infer[(\textbf{weak})]{t \leq_r s}{t \leq_q s \qquad r \leq q}
        \qquad \quad
        \infer[(\textbf{join})]{t \leq_{\vee q_i} s}{\forall i \leq n. \> t \leq_{q_i} s}
\]
\\[2pt]
\[
        \begin{prooftree}
                \hypo { \Gamma, x : \typeA \vljud t \leq_q s }
                \infer1[]{\Gamma \vljud t[u/x]  \leq_q s[u/x]}
        \end{prooftree}
        \hspace{2cm}
        \begin{prooftree}
                    \hypo{ \Gamma \vljud t \leq_q s}
                    \hypo{ \Gamma \simeq_\pi \Delta }
                    \infer2[]{ \Delta \vljud t \leq_q s}
        \end{prooftree}
        \hspace{2cm}
        \begin{prooftree}
                \hypo{\forall i \leq n. \>  t_i \leq_{q_i} s_i}
                \infer1[]{f(t_1,\dots,t_n) \leq_{\wedge q_i} f(s_1,\dots,s_n)}
        \end{prooftree}
\]
\caption{$\Quant$-axiomatic schema of the presheaf $P$ (above dotted line) and the
$\Quant$-(in)equational system of multi-sorted universal algebra (below the dotted
line).}
\label{fig:theo_rules}
\end{figure}
As the next step in establishing that $\mathscr{L}$ is left adjoint to the
inclusion $i$, consider now a natural transformation $\alpha : P \to i(Q)$. We
wish to build another one $\hat \alpha : \mathscr{L}(P) \to Q$ such that
$i(\hat \alpha) \comp \eta = \alpha$. Interestingly the latter equation,
naturality, and the fact that $Q$ preserves products, entails the uniqueness of
this construction.  Specifically $\hat \alpha$ must be defined by,
\[
        \hat \alpha ([c]) = \alpha (c)
        \qquad
        \hat \alpha([f(t)]) = Q(f) (\hat \alpha ([t]))
        \qquad
        \hat \alpha ([(t,s)]) = \pv{Q(\pi_1)}{Q(\pi_2)}^{-1}
        (\hat \alpha ([t]), \hat \alpha ([s]))
\]
Indeed, the fact that the first and second equation must hold is direct.  As for
the third equation, observe that due to naturality and the fact that $Q$
preserves products the following diagram must commute.
\[
        \begin{tikzcd}[column sep=40pt, row sep=30pt]
                \mathscr{L}(P)(X \times Y)
                \arrow[r, "\alpha_{X \times Y}"]
                \arrow[d, "\pv{\mathscr{L}(\pi_1)}{\mathscr{L}(\pi_2)}"']
                &
                Q(X \times Y)
                \arrow[d, "\pv{Q(\pi_1)}{Q(\pi_2)}"']
                \\
                \mathscr{L}(P)(X) \times
                \mathscr{L}(P)(Y)
                \arrow[r, "\alpha_X \times \alpha_Y"']
                &
                Q(X \times Y)
                \arrow[u, bend right=40, "\pv{Q(\pi_1)}{Q(\pi_2)}^{-1}"']
        \end{tikzcd}
\]
The fact that $\hat \alpha$ is well-defined (\ie\ that it respects the
previous equivalence relation) is a corollary of the following result. If $x_1
: \typeA_1, \dots, x_n : \typeA_n \vljud t \leq_q s$ then,
\[
        \forall (- \vljud u_i : \typeA_i)_{1 \leq i \leq n} . \> \>
        q \leq a( \hat \alpha ([t[u_1/x_1]\dots[u_n/x_n]]) , 
        \hat  \alpha ([s[u_1/x_1]\dots[u_n/x_n]]) ) 
\]
This is proved by a straightforward induction over the depth of proof
derivation trees. The fact each each component $\hat \alpha_X$ is a
$\Quant$-functor follows from the same result.  The final step in our quest is
to prove that the mapping $\alpha \mapsto \hat \alpha$ just defined is a
$\Quant$-functor. This follows from the result,
\[
        \forall (- \vljud t : \typeA). \> a(\alpha,\beta) \leq a(\hat
        \alpha([t]), \hat \beta([t]))
\]
which is straightforward to prove via induction on the structure of $t$.

We have thus established that $\mathscr{L}$ is left adjoint to $i$.  Let us
recall from the literature why one can now deduce that
$[\catC^\cop,\VCatSe]_{\fp}$  has all the properties in need to be an instance
of Definition~\ref{defn:auto_coprod}.  First since it is a reflective
subcategory of $[\catC^\cop,\VCatSe]$ it has $\VCatSe$-enriched
(co)limits~\cite{cats}. Second, recall the assumption that the tensor in
$\catC$ distributes over coproducts. This entails that the exponential
$\multimap$ in $[\catC^\cop,\VCatSe]$ restricts to
$[\catC^\cop,\VCatSe]_{\fp}$, specifically,
\begin{align*}
        P \multimap Q (Y + Z) & \cong \int_X \VCat(P(X), Q((Y + Z) \otimes X))
        \\
        & 
        \cong \int_X \VCat(P(X), Q((Y \otimes X) + (Z \otimes X))
        \\
        &
        \cong \int_X \VCat(P(X), Q(Y \otimes X) \times Q (Z \otimes X))
        \\
        & 
        \cong \int_X \VCat(P(X), Q(Y \otimes X))
        \times
        \int_X \VCat(P(X),Q (Z \otimes X))
        \\
        &
        \cong P \multimap Q (Y) \times P \multimap Q (Z) 
\end{align*}
The $\VCatSe$-autonomous structure of $[\catC^\cop,\VCatSe]_{\fp}$ is then
established via Day's reflection theorem~\cite{day72,malherbe13,lack14}:  the
tensor of $P$ and $Q$ is $\mathscr{L}(i(P) \otimes i(Q))$ and the exponential
is inherited from the supercategory.

\section{Probabilistic and Quantum Programming}
\label{sec:prob_quantum}
\subsection{Probabilistic programming}
        We now present two concrete models of our calculus, in the context
        of probabilistic and quantum programming. We start with the
        probabilistic case.  Consider the category $\Ban$ of Banach spaces and
        linear \emph{contractions} -- a core category in functional analysis
        and the basis of a semantic framework for (higher-order) probabilistic
        programming~(see
        \cite{dahlqvist19,dahlqvist22,dahlqvist23,dahlqvist23b}).  It was
        already shown in the~\emph{op. cit.} that it is $\Met$-autonomous. Note
        as well that it has binary coproducts, given by the direct sum of
        vector spaces and the $\ell_1$-norm. Thus it remains to show that such
        coproducts are enriched over the Cartesian structure of $\Met$ (recall
        Section~\ref{sec:sem}). This is somewhat folklore, but we sketch a
        proof below for completeness.  First, given a linear map $T : V \to W$
        between Banach spaces we define,
        \[
         \norm{T} = \bigvee \left \{ \norm{T(v)} \mid v \in V, \norm{v} = 1 \right \}
        \]
        Linear contractions are precisely those linear maps $T$ such that $\norm{T}
        \leq 1$, and the metric $d(T,S)$ between two contractions $T$ and $S$ is set
        as $\norm{ T - S}$. Given $T : V \to W$ and $S : U \to W$ their co-pairing
        $[T,S] : V \oplus U \to W$ is defined by $[T,S](v,u) = T(v) + S(u)$. The fact
        that $[T,S]$ is contractive follows from the inequation $\norm{[T,S]} \leq \max
        \{ \norm{T}, \norm{S} \}$ for all linear maps $T,S$. This becomes straightforward to
        prove when one notices that every unitary vector $(v,u) \in V \oplus U$ can be
        rewritten as,
\[
        \left (\norm{v} \frac{1}{\norm{v}} v, \norm{u} \frac{1}{\norm{u}} u \right )
        \hspace{1cm}
        \norm{v} + \norm{u} = 1
\]
The fact that binary coproducts in $\Ban$ are enriched over the Cartesian
structure of $\Met$ then follows rather directly,
\begin{align*}
        d([T,S] , [T',S']) 
        & = 
        \norm{ [T,S] - [T',S'] }
        \\
        & = 
        \norm{ [T - T', S - S'] }
        \\
        & \leq
        \max \{ \norm{T - T'}, \norm{S - S'} \}
        \\
        & =
        \max \{ d(T,T'), d(S,S') \}
\end{align*}
In order to introduce probabilities to this model, we will involve the notion
of a \emph{measure} which we briefly describe next  (see \eg\ \cite[Chapter
10]{aliprantis06} or \cite[Chapter 2]{panangaden09} for a thorough account).
\begin{definition} For a measurable space $(X,\Sigma_X)$ a measure is a
        function $\mu : \Sigma_X \to \Reals$ such that $\mu(\emptyset) = 0$ and
        moreover it is $\sigma$-additive, \ie\ 
        \[
                \mu \left (\bigcup_{i =1}^{\infty} U_i \right ) = \sum_{i = 1}^{\infty}
                \mu(U_i) 
        \] 
where $(U_i)_{i \in \omega}$ is any family of pairwise disjoint measurable
sets.  A measure $\mu$ is called \emph{positive} if $\mu(U) \geq 0$ for all
measurable sets $U$ and a \emph{distribution} if furthermore $\mu(X) =1$.  
\end{definition}
It follows from $\sigma$-additivity that a 
measure $\mu$ satisfies,
\[
\mu \left (\bigcup_{i \in \omega} U_i \right ) = \bigvee_{i \in \omega} \mu(U_i)
\]
for any increasing
sequence $(U_i)_{i \in \omega}$ of measurable subsets.  
Next, for a measurable space $X$ the set of measures $M(X)$ forms a vector space
via pointwise extension. It also forms a Banach space when equipped with the
total variation norm,
\[
        \lVert \mu \rVert = 
        \bigvee \left \{ \sum_{i = 1}^n \, \lVert \mu(U_i) \rVert \mid
             \{ U_1, \dots, U_n \} \text{ is a measurable partition}
        \right \}
\]
In many cases such a norm can be difficult to compute, but note that there are
helpful results in this regard.  For example, a useful fact about computing
norms is that $\norm{\mu} = \mu^{+}(X) + \mu^{-}(X)$ where $\mu^{+}$ and
$\mu^{-}$ are the positive and negative parts of $\mu$ respectively (see
details in~\cite[Section 8.2. and Section 10.10]{aliprantis06}). We will use
this result later on.

Now, we proceed by presenting a metric $\lambda$-theory
(Definition~\ref{defn:theory}) on which to reason about predicates and random
walks. This connects us directly to our concrete motivation in
Section~\ref{sec:intro}. Our only ground type will be $\mathtt{real}$ to
represent measures over real numbers -- \ie\ we set $\sem{\mathtt{real}}$ to be
the space $M(\Reals)$ of measures over the real line. Recall that the monoidal
unit of $\Ban$ is $\Reals$. Concerning operations we take a pre-determined set
of predicates $p : \mathtt{real} \to \typeI \oplus \typeI$ whose interpretation
takes the form $\sem{p}(\mu) = (\mu(U), \mu(\overline{U})) \in \Reals \oplus
\Reals$ for some measurable subset of $U \subseteq \Reals$.  Intuitively $U
\subseteq \Reals$ corresponds to the subspace in which the predicate is
supposed to hold. We also take a pre-determined set of actions $a : \typeI \to
(\typeA \multimap \typeA)$ and a pre-determined set of measures $m : \typeI \to
\mathtt{real}$ whose interpretation takes no particular form.  Finally we also
include addition $+ : \mathtt{real},\mathtt{real} \to \mathtt{real}$, whose
interpretation is given by $\mu \otimes \nu \mapsto +_\ast(\mu \otimes \nu)$
where $+_\ast$ is the pushforward measure construction of $+$ (see further
details in~\cite{dahlqvist22,dahlqvist23}).

We will now briefly illustrate \emph{quantitative} reasoning in this probabilistic
framework. In particular we will highlight the synergies that arise by the
newfound possibility of using both quantitative syntactic and semantic
machinery in tandem. As mentioned in Section~\ref{sec:intro}, we will focus
specifically on (Cauchy sequences of) random walks.  
First, given a measure $m$ and actions $a,b$
consider the following `abstract' Bernoulli trial,
\[
        p : \mathtt{real} \multimap \typeI \oplus \typeI
        \vljud 
        \underbrace{\ccelim{p(m(\ast))}{a(x)}{b(y)}}_{\mathtt{bern}(p)} 
        : \typeA \multimap \typeA
\]
Note that this is an abstract version of a step of a random walk as described
in Section~\ref{sec:intro}.  Note as well that if the metric equation
$p_1(m(\ast)) =_\epsilon p_2(m(\ast))$ holds for two predicates $p_1, p_2 :
\mathtt{real} \to \typeI \oplus \typeI$ then the equation
$\mathtt{bern}(\lambda x. p_1(x)) =_\epsilon \mathtt{bern}(\lambda x.  p_2(x))$
must hold as well (as per our equational system, in particular the
$\beta$-equation w.r.t.~$\lambda$-abstractions). Such is useful to approximate
Bernoulli trials that may be hard to compute as illustrated by the following
examples.
\begin{example}[Predicates and Cauchy sequences]
        Take a measure $m$ and the predicate,
        \[
                x : \mathtt{real} \vljud
                p_{\frac{1}{2}\sqrt{2}}(x) : \typeI \oplus \typeI
        \]
        that returns true if $x < \frac{1}{2}\sqrt{2}$ and false otherwise.
        Given the irrationality of $\frac{1}{2}\sqrt{2}$ it is natural to
        consider successive approximations $(-) \vljud p_{q_n}(m(\ast))
        :  \typeI \oplus \typeI$ $(n \in \Nats)$ in which the
        condition $x < \frac{1}{2}{\sqrt{2}}$ is replaced by $x < q_n$ for
        $q_n$ a rational number. We show next how our framework makes this idea
        precise. Take a sequence of rational numbers $(q_n)_{n \in \Nats}$ that
        converges to $\frac{1}{2}\sqrt{2}$ from below. We then postulate as
        axioms in our deductive system that $(p_{q_n}(m(\ast)))_{n \in \Nats}$ is a
        Cauchy sequence and furthermore that it converges to $p_{\frac{1}{2}
        \sqrt{2}}(m(\ast))$.  Such is asserted precisely by setting,
        \begin{equation}
                \label{eq:sound}
                \begin{cases}
                \forall \epsilon > 0. \, \exists k \in \Nats.
                \, \forall n \geq k. \, p_{q_n}(m(\ast)) =_\epsilon p_{q_{n+1}} (m(\ast))
                & \text{(Cauchy sequence)}
                \\
                \forall \epsilon > 0. \, \exists k \in \Nats.
                \, \forall n \geq k. \, p_{q_n}(m(\ast)) 
                =_\epsilon p_{\frac{1}{2} \sqrt{2}} (m(\ast))
                & \text{(Convergence)}
                \end{cases}
        \end{equation}
        for appropriate choices of $k$ (which in our context is irrelevant to
        detail). The next step is to prove that this axiomatics is sound, \ie\
        that such equations hold in $\Ban$ (meaning that we will now proceed
        with a semantic approach). In particular we make the following
        reasoning,
        \begin{flalign*}
                \sem{p_{\frac{1}{2} \sqrt{2}}(x)}(\mu)
                & = \left (\mu \left (-\infty, \frac{1}{2} \sqrt{2} \right ), 
                \mu(\Reals) - \mu \left (-\infty, \frac{1}{2}\sqrt{2} \right ) \right )
                \\
                & = \left (\mu \left (\bigcup_{n \in \Nats} 
                        \left (-\infty, q_n \right )\right ), 
                \mu(\Reals) - \mu \left (-\infty, \frac{1}{2}\sqrt{2} \right ) \right )
                & 
                \left \{(q_n)_{n \in \Nats} 
                \nearrow \frac{1}{2}\sqrt{2} \right \}
                \\
                & = \left (\sup_{n \in \Nats} \mu \left ( 
                        \left (-\infty, q_n \right )\right ), 
                \mu(\Reals) - \mu \left (-\infty, \frac{1}{2}\sqrt{2} \right ) \right )
                & 
                \left \{ \text{Measure properties} \right \}
                \\
                & = \left (\lim_{n \to \infty} \mu \left ( 
                        \left (-\infty,  q_n \right )\right ), 
                \mu(\Reals) - \mu \left (-\infty, \frac{1}{2}\sqrt{2} \right ) \right )
                & 
                \left \{ \text{Limits coincide with sup. of inc. seq.} \right \}                         \\
                & = \left (\lim_{n \to \infty} \mu \left ( 
                        \left (-\infty, q_n \right )\right ), 
                \lim_{n \to \infty}
                \mu \overline{\left (-\infty, q_n \right )} \right )
                & 
                \left \{ \text{Measure properties} \right \}
                \\
                & = \lim_{n \to \infty} \left (\mu \left ( 
                        \left (-\infty,  q_n \right )\right ), 
                \mu \overline{\left (-\infty, q_n \right )} \right )
                & 
                \\
                & = \lim_{n \to \infty}
                \sem{p_{q_n}(x)}(\mu)
                &
        \end{flalign*}
        Thus we are indeed able to approximate with arbitrary precision a step
        of a random walk $\mathtt{bern}(\lambda x.  \, p_{\frac{1}{2} \sqrt{2}}
        (x) )$, using successive approximations $p_{q_n}$ of the predicate
        $p_{\frac{1}{2} \sqrt{2}}(x)$.  In the next example we further
        capitalise on such approximations, now formulated precisely, to reason
        about approximations of (bounded) random walks.
        \end{example}

\begin{example}[Random walk approximations] 
        We consider the $\lambda$-term, which operationally speaking sequences
        $k$ terms given as input,
        \[
                (-) \vljud \underbrace{\lambda x_1. \, \dots \, x_k. \, y. \,
                x_1 (\dots (x_k(y)) \dots)}_{\mathtt{sequence_k}}
        \]
        Also
        given a predicate $p : \mathtt{real} \to \typeI \oplus \typeI$, take
        $(-) \vljud \mathtt{sequence_k} \>
        \mathtt{bern}(\lambda x. \, p(x)) \dots \, \mathtt{bern}(\lambda x. \,
        p(x)) : \typeA \multimap \typeA$ which intuitively represents an
        \emph{abstract random walk of $k$-steps}. In order to keep our notation
        simple, we abbreviate this last term to $\mathtt{rwalk}(\lambda x. \,
        p(x))$. Now, it follows from our system, specifically from metric
        congruence and the previous example, that if we have $p_1(m(\ast)) =_\epsilon
        p_2(m(\ast))$ for two predicates $p_1$ and $p_2$ and a measure $m$ then,
        \[
                \mathtt{rwalk}(\lambda x. p_1(x)) =_{k \cdot \epsilon}
                \mathtt{rwalk}(\lambda x. p_2(x)) 
        \]
        In particular, from the previous example we deduce that $\mathtt{rwalk}
        (\lambda x. \, p_{q_n} (x))$ is a Cauchy sequence that converges to
        $\mathtt{rwalk}(\lambda x.  \, p_{\frac{1}{2} \sqrt{2}} (x) )$. In
        other words, the approximations obtained in the previous example
        propagate to the corresponding random walks if bounded.
        \end{example}
        \begin{example}
        As a final illustration of the synergy between syntax and semantics
        that our framework provides, suppose now that the actions $a,b : \typeI
        \to \typeA \multimap \typeA$ involved in $\mathtt{bern}(\lambda x. \,
        p(x))$ are concrete jumps on the real line. Specifically, we axiomatise,
        \[
                a(\ast) = \lambda z. \, +(z, \mathtt{unif}(0,1)(\ast))
                \hspace{1cm}
                b(\ast) = \lambda z. \, +(z, \mathtt{unif}(-1,0)(\ast))
        \]
        where $\mathtt{unif}(0,1) \in M(\Reals)$ is the uniform distribution on
        the interval $[0,1]$ and analogously for $\mathtt{unif}(-1,0)$.
        Operationally $a$ corresponds to a jump to the right with magnitude
        between $0$ and $1$, and analogously for $b$.  Suppose we have another
        action $c : \typeI \to (\mathtt{real} \multimap \mathtt{real})$ whose
        interpretation is that of $a$ except for the fact that
        $\mathtt{unif}(0,1)$ is replaced by $\mathtt{unif}(0,1+\delta)$.  What
        will be the effect on the random walk when replacing $a$ by $c$?  Our
        approach starts \emph{semantically}, by computing the norm
        $\norm{\mathtt{unif}(0,1) - \mathtt{unif}(0,1+\delta)}$: per our previous
        remarks about measures it decomposes into,
        \begin{flalign*}
                (\mathtt{unif}(0,1) - \mathtt{unif}(0,1+\delta))^+ (\Reals)
                +
                (\mathtt{unif}(0,1) - \mathtt{unif}(0,1+\delta))^- (\Reals)
        \end{flalign*}
        then proceed by computing the left-hand side of this addition,
        \begin{flalign*} 
        & \, (\mathtt{unif}(0,1) - \mathtt{unif}(0,1+\delta))^+ (\Reals)
        &
        \\
        & = 
        \bigvee \{ \mathtt{unif}(0,1)(U) - \mathtt{unif}(0,1 + \delta)(U)
        \mid U \subseteq \Reals \}
        \\
        & 
        =
        \bigvee \{ \mathtt{unif}(0,1)(U \cap [0,1]) 
        - \mathtt{unif}(0,1 + \delta)(U \cap [0,1]) 
        - \mathtt{unif}(0,1 + \delta)(U \cap (1,1 + \delta]) 
        \mid U \subseteq \Reals \}
        \\
        &
        = \bigvee \left \{ \left (1 - \frac{1}{1+\delta} \right ) 
                \mathtt{unif}(0,1)(U \cap [0,1]) 
        - \mathtt{unif}(0,1 + \delta)(U \cap (1,1 + \delta]) 
        \mid U \subseteq \Reals \right \}
        \\
        & = 1 - \frac{1}{1+\delta} 
        \end{flalign*}
        It follows from an analogous reasoning that the right-hand side of the
        addition will be $\frac{\delta}{1 + \delta}$ and therefore the norm
        will be exactly $2 \cdot {\frac{\delta}{1 + \delta}}$. Thus
        we can soundly axiomatise the metric equation
        $\mathtt{unif}(0,1)(\ast) =_{2 \cdot {\frac{\delta}{1 + \delta}}}
        \mathtt{unif}(0,1+\delta)(\ast)$, which by metric congruence
        entails $a(\ast) =_{2 \cdot {\frac{\delta}{1 + \delta}}} c(\ast)$.
        We continue by reasoning \emph{syntactically}, now about a step of 
        the previous abstract random walk.
        \begin{flalign*}
               & \, \ccelim{p(m(\ast))}{a(x)}{b(y)}
               &
               \\
               & =_0
               \ccelim{p(m(\ast))}{\ielim{x}{a}(\ast)}{b(y)}
               &
               \\
               & =_{2 \cdot {\frac{\delta}{1 + \delta}}}
               \ccelim{p(m(\ast))}{\ielim{x}{c}(\ast)}{b(y)}
               &
               \\
               & =_0
               \ccelim{p(m(\ast))}{c(x)}{b(y)}
        \end{flalign*}
        Thus again by metric congruence
        if $\mathtt{rwalk}(\lambda x. p(x))$ is the random walk that
        involves action $a$ and $\mathtt{rwalk'}(\lambda x. p(x))$ the random
        walk that involves action $c$ we prove from our framework the metric
        equation,
        \[
                \mathtt{rwalk}(\lambda x. p(x)) =_{2k \cdot \frac{\delta}{1 + \delta}}
                \mathtt{rwalk'}(\lambda x. p(x)) 
        \]
        which will converge to $0$ as $\delta$ tends to $0$.
\end{example}

\subsection{Quantum programming}

We will now construct a $\Met$-enriched category for (higher-order) quantum programming.
Let $\CPTP$ be the category whose objects are the natural numbers $n \in \Nats$
and morphisms $n \to m$ are the \emph{quantum channels} $\Phi : \Complex^{n
\times n} \to \Complex^{m \times m}$, \ie\ completely positive
\emph{trace-preserving} operators (see a thorough account of such maps in
\eg~\cite{watrous18}).  This category is a $\Met$-enriched symmetric monoidal
category~\cite{dahlqvist23}.  Unfortunately it does not meet the pre-requisite
of having binary coproducts.  The latter however in some sense ``live hidden''
in the category and what is required is to ``draw them out" in a suitable way.
Our strategy for achieving this is based on the observation that while $\CPTP$
does not have coproducts it has enough projections (\ie\ idempotents) to carve
coproducts out of them.  Technically this is obtained via the so-called
\emph{Karoubi envelope} $\Kar{-}$~\cite{borceux94} (also known as the
idempotent split completion).

\begin{definition}
        The Karoubi envelope $\Kar{\CPTP}$ of $\CPTP$ has as objects pairs $(n,
        P)$ with $P : \Complex^{n \times n} \to \Complex^{n \times n}$ an
        idempotent quantum channel. A morphism $\Phi : (n,
        P) \to (m,Q)$ in $\Kar{\CPTP}$ is a $\CPTP$-morphism $n \to m$  such that
        $Q \comp \Phi \comp P = \Phi$. The identity of $(n ,P)$ is precisely
        $P$ and composition is inherited from $\CPTP$.
\end{definition}
The category $\Kar{\CPTP}$ inherits useful structure from $\CPTP$.  First since
$\CPTP$ is already $\Met$-enriched, the category $\Kar{\CPTP}$ becomes
$\Met$-enriched as well by stipulating the metric in $\Kar{\CPTP}(n,m)$ to be
the respective restriction of the metric in $\CPTP(n,m)$. Second $\Kar{\CPTP}$
inherits symmetric monoidal structure from $\CPTP$ (see for
example~\cite[Proposition 6.5.9]{borceux94}), where in particular the tensor of
$(n,P) \otimes (m,Q)$ will be $(n m , P \otimes Q)$ and the tensoring of
morphisms is as in $\CPTP$. The category $\Kar{\CPTP}$ is thus even a
$\Met$-enriched symmetric monoidal category. Our next step is to show that it is
also $\Met$-co-Cartesian.
\begin{remark}
        The idea of using the Karoubi envelope to build a co-Cartesian category
        for quantum computing was first proposed in~\cite{selinger08b} -- where
        the author applies (a variant of) this construction to the category
        $\CPM$ of completely positive operators.  This contrasts with our
        $\Met$-enriched setting, which imposes \emph{trace-preservation} to the
        underlying morphisms (thus our adoption of $\CPTP$ rather than $\CPM$;
        see~\cite{dahlqvist23} for more details, but note for example that
        trace-preservation ensures that the substitution rule from
        Figure~\ref{fig:equations-in-context-cond} is sound). 

        Another natural candidate for a co-Cartesian quantum model is the
        coproduct completion of $\CPTP$ (see details of this construction in
        \eg~\cite{adamek20}). However a close inspection of the resulting
        category reveals that it has ``too few'' morphisms. For example, while
        the qubit measurement operation is supposed to be interpreted as a map
        of type $2 \to 1 + 1$ (see~\cite{selinger08b,dahlqvist23}), the
        category only harbours maps $2 \to 1 + 1$ either of the form $A \mapsto
        (\mathop{Tr} A,0)$ or $A \mapsto (0,\mathop{Tr} A)$ where $\mathop{Tr}$
        is the trace operation.  Yet another option is to first take the
        bicompletion of $\CPM$~\cite{selinger07} and then \emph{restrict} it to
        trace-preserving maps.  The issue is that this restriction does not
        have an obvious $\Met$-enrichment \emph{a priori}. An interesting
        question, which we do not pursue here, is whether this restriction can
        be equipped with a suitable $\Met$-enrichment via our category
        $\Kar{\CPTP}$ and the embedding results presented in~\cite{heunen13}.
\end{remark}

Let us thus show that $\Kar{\CPTP}$ has binary coproducts. We start by
focussing on the notion of co-pairing.  Take a quantum channel $\Phi : n \to m$
and consider $[\Phi ,0] : n+o \to m$ defined by,
\[
 \begin{pmatrix}
         A & B 
         \\
         C & D
 \end{pmatrix} 
 \quad
 \longmapsto 
 \quad
 \Phi(A)
\]
By an appeal to Kraus' representation theorem~\cite[Theorem 2.22 and Corollary
2.27]{watrous18} one easily shows that this map is completely positive (but not
trace-preserving), as detailed next.  
\begin{flalign*}
        [\Phi, 0]
        \begin{pmatrix}
         A & B 
         \\
         C & D
        \end{pmatrix} 
        & =
        \Phi(A)
        \\
        &
        = 
        \sum_i M_i A M^\dagger_i
        & \big \{ \text{Kraus' representation theorem} \big \}
        \\
        &
        = \sum_i 
        \begin{pmatrix}
        M_i & 0 
        \end{pmatrix} 
        \begin{pmatrix}
        A & B
         \\
        C & D
        \end{pmatrix} 
        \begin{pmatrix}
        M_i^\dagger \\
        0 
        \end{pmatrix} 
\end{flalign*}
where $(-)^\dagger$ is the usual adjoint operation.
Observe the existence of an analogous operator $[0,\Phi] : o + n \to m$.
Recall as well that Kraus' representation theorem entails that the addition of
completely-positive operators is also completely positive. This means that
given quantum channels $\Phi : n \to m$ and $\Psi : o \to m$ their
`co-pairing',
\[
        [\Phi, \Psi] := [\Phi, 0] + [0, \Psi]
\]
will be completely positive as well. What is more, it will be trace-preserving,
for it is easy to see that the respective Kraus' operators add up to the block
identity matrix.  Next, consider the quantum channels $e_l : n \to n +m$ and
$e_r : m \to n + m$,
\[
        A \quad \longmapsto \quad
        \begin{pmatrix}
                A & 0 
                \\
                0 & 0
        \end{pmatrix}
        \qquad
        \qquad
        D \quad \longmapsto \quad
        \begin{pmatrix}
                0 & 0 
                \\
                0 & D
        \end{pmatrix}
\]
which have Kraus' operators $\begin{pmatrix} I \\ 0 \end{pmatrix}$ and
$\begin{pmatrix} 0 \\ I \end{pmatrix}$ respectively.  We define the `direct
sum' of quantum channels $\Phi \oplus \Psi$ by setting $\Phi \oplus \Psi = [e_l
\comp \Phi, e_r \comp \Psi]$, a quantum channel by construction. Observe that
$\Phi \oplus \Psi$ is a particular case of a so-called
\emph{quantum-to-classical} channel~\cite[Definition 2.35]{watrous18}, and
specifically $\id \oplus \id$ is a block matrix generalisation of the
\emph{completely dephasing} channel~\cite[Equation (2.162)]{watrous18}. The
crucial observation then is that idempotents in $\CPTP$ of the form $\Phi
\oplus \Psi$ give rise to coproducts in $\Kar{\CPTP}$, as we will show in the
following proposition.
\begin{proposition}
        The category $\Kar{\CPTP}$ has binary coproducts.
\end{proposition}

\begin{proof}
        Given $(n,P)$ and $(m,Q)$ consider the object $(n + m , P \oplus Q)$.
        The injections $\inl$ and $\inr$ are defined as $\inl := (P \oplus Q)
        \comp e_l \comp P$ and $\inr := (P \oplus Q) \comp e_r \comp Q$. Given
        $\Kar{\CPTP}$-morphisms $\Phi : (n, P) \to (o, R)$ and $\Psi : (m, Q)
        \to (o,R)$ their co-pairing is defined as above.
        The fact that $[\Phi,\Psi] \comp \inl = \Phi$ follows from the
        idempotency of $P$ and the fact that $\Phi \comp P = \Phi$. The same
        reasoning applies to $[\Phi,\Psi] \comp \inr = \Psi$. In order to prove
        unicity consider a suitably typed operator $\Upsilon$ such that $\Upsilon
        \comp \inl = \Phi$ and $\Upsilon \comp \inr = \Psi$. Then we reason,
        \begin{flalign*}
        \Upsilon
        \begin{pmatrix}
         A & B 
         \\
         C & D
        \end{pmatrix} 
        & =
        \Upsilon
        \begin{pmatrix}
         P(A) & 0 
         \\
         0 & Q(D)
        \end{pmatrix} 
        & \big \{ \text{$\Upsilon$ is a $\Kar{\CPTP}$-morphism} \big \}
        \\
        & =
        \Upsilon
        (\inl(A) + \inr(D))
        \\
        & = 
        \Upsilon
        (\inl(A)) + \Upsilon(\inr(D))
        \\
        & = 
        \Phi(A) + \Psi(D)
        \\
        &
        = 
        [\Phi,\Psi]
        \begin{pmatrix}
         A & B 
         \\
         C & D
        \end{pmatrix} 
\end{flalign*}
\end{proof}

In order to prove that the binary coproducts of $\Kar{\CPTP}$ are enriched over
$\Met$'s Cartesian structure, we will need the following result.
\begin{proposition}
        \label{prop:dist}
        The category $\Kar{\CPTP}$ is strictly distributive.
\end{proposition}

\begin{proof}
Take objects $(n,P)$, $(m,Q)$, $(o,R)$, and note that the following equations
are sound.
\[
        \begin{cases}
                \big ((n,P) + (m,Q) \big ) \otimes (o,R) = (no + mo, (P
                \oplus Q) \otimes R)
                \\
                \big (n,P) \otimes (o,R) + (m,Q) \otimes (o,R) = (no + mo, 
                        (P \otimes R) \oplus (Q \otimes R))
        \end{cases}
\]
They tell that the two relevant composite objects are actually the same if $(P \oplus Q)
\otimes R = (P \otimes R) \oplus (Q \otimes R)$. So we reason,
        \begin{align*}
        (P \oplus Q) \otimes R
        \left (
        \begin{pmatrix}
         A & B 
         \\
         C  & D 
        \end{pmatrix} 
        \otimes E
        \right )
        & =
        \left (
        \begin{pmatrix}
         P(A) &  0
         \\
         0 & Q(D)
        \end{pmatrix} 
        \otimes R(E)
        \right )
        \\
        & 
        =
        \begin{pmatrix}
         P(A) \otimes R(E) &  0
         \\
         0 & Q(D) \otimes R(E)
        \end{pmatrix} 
        \\
        & 
        = (P\otimes R \oplus Q \otimes R)
        \begin{pmatrix}
         A \otimes E & B \otimes E
         \\
         C \otimes E & D \otimes E
        \end{pmatrix} 
        \\
        & 
        = (P\otimes R \oplus Q \otimes R)
        \left ( \begin{pmatrix}
         A  & B
         \\
         C & D
        \end{pmatrix} 
        \otimes E
        \right )
        \end{align*}
        which indeed establishes an equality between the two composite objects.  The
        final step then is to prove that the equation $[\inl \otimes \id, \inr
        \otimes \id] = \id$ holds, which by unicity reduces to the equations,
        \[
                \begin{cases}
                        \inl \otimes \id = \inl \\
                        \inr \otimes \id = \inr
                \end{cases}
        \]
        and whose proof is direct.
\end{proof}

We will now focus on the $\Met$-enriched setting. Specifically we start by
further detailing the $\Met$-enriched structure of $\Kar{\CPTP}$ and then prove
that the latter's binary coproducts are enriched over $\Met$'s Cartesian
structure. First recall that given a matrix $A \in \Complex^{n\times n}$ its
trace norm $\norm{A}$ (or Schatten 1-norm) is defined by $\norm{A} =
\mathop{Tr}\sqrt{A^\dagger A}$~\cite{watrous18}. This induces the usual norm
for operators $T : \Complex^{n \times n} \to \Complex^{m \times m}$,
\[
        \norm{T} = \bigvee \left \{ \norm{A} \mid A \in \Complex^{n \times n}, \norm{A} \leq 1 
        \right \}
\]
However as explained in~\cite{watrous18}, this norm is not stable under tensoring.
This makes it impossible to enrich the monoidal structure of $\CPTP$ via
such a norm, and motivates the adoption of the \emph{diamond norm} instead,
\[
        \norm{T}_\diamond = \norm{T \otimes \id_n}
\]
which indeed satisfies the pre-requisites for making $\CPTP$ $\Met$-symmetric
monoidal~\cite{dahlqvist23}. Specifically we set the distance between quantum
channels $\Phi$ and $\Psi$ in $\CPTP$ to be $\norm{\Phi - \Psi}_\diamond$.
Next, note that the morphisms of $\Kar{\CPTP}$ are a subset of the
$\CPTP$-morphisms and that the symmetric monoidal structure of $\Kar{\CPTP}$ is
inherited from $\CPTP$. It is then clear that $\Kar{\CPTP}$ becomes
$\Met$-symmetric monoidal when one adopts the respective restriction of the
metric structure in $\CPTP$. As for binary coproducts and their Cartesian
enrichment, we will first show that the inequation $\norm{[T,S]} \leq \max \{
\norm{T},\norm{S} \}$ holds for all operators $T : \Complex^{n \times n} \to
\Complex^{o \times o}$ and $S : \Complex^{m \times m} \to \Complex^{o \otimes
o}$. Then we will extend this property to the setting of the diamond norm.
\begin{flalign*}
        \norm{[T,S]}
        & =
        \bigvee \left \{ \norm{T(A) + S(D)} \mid 
                \left \lVert
                      \begin{pmatrix}
                              A & B
                              \\
                              C & D
                      \end{pmatrix}
                \right \rVert \leq 1 
        \right \}
        \\
        & \leq
        \bigvee \left \{ \norm{T(A) + S(D)} \mid 
                \left \lVert
                      \begin{pmatrix}
                              A & 0
                              \\
                              0 & D
                      \end{pmatrix}
                \right \rVert \leq 1 
        \right \}
        & \big \{ \text{$\norm{\id \oplus \id} = 1$} \big \}
        \\
        & =
        \bigvee \left \{ \norm{T(A) + S(D)} \mid 
                \left \lVert
                      \begin{pmatrix}
                              A & 0
                              \\
                              0 & D
                      \end{pmatrix}
                \right \rVert = 1 
        \right \}
        & 
        \\
        & =
        \bigvee \left \{ \norm{T(A) + S(D)} \mid 
                \norm{A} + \norm{D} = 1
        \right \}
        &  
        \big \{ \text{see~\cite[1.148]{watrous18}} \big \}
        \\
        & \leq
        \bigvee \left \{ \norm{T(A)} + \norm{S(D)} \mid 
                \norm{A} + \norm{D} = 1
        \right \}
        & 
        \\
        & =
        \bigvee \left \{ \norm{A}\norm{T(\sfrac{1}{\norm{A}}A)} + 
                \norm{D}\norm{S(\sfrac{1}{\norm{D}}D)} \mid 
                \norm{A} + \norm{D} = 1
        \right \}
        & 
        \\
        & \leq
        \bigvee \left \{ \max \{ \norm{T(\sfrac{1}{\norm{A}}A)},
                \norm{S(\sfrac{1}{\norm{D}}D)} \} \mid 
                \norm{A} + \norm{D} = 1
        \right \}
        & 
        \\
        &
        \leq \max \{ \norm{T},\norm{S} \} 
\end{flalign*}
This property extends to the diamond norm via the following reasoning.
\begin{flalign*}
        \norm{[T,S]}_\diamond 
        & =
        \norm{[T,S] \otimes \id_{n+m}}
        \\
        & = 
        \norm{[T \otimes \id_{n+m},S \otimes \id_{n+m}] }
        & \text{\{Proposition~\ref{prop:dist}\}}
        \\
        & \leq
        \max \{ \norm{T \otimes \id_{n+m}}, \norm{S \otimes \id_{n+m}} \}
        \\
        & = 
        \max \{ \norm{T \otimes \id_{n}}, \norm{S \otimes \id_{m}} \}
        &
        \text{\{\cite[Theorem 3.46]{watrous18}\}}
        \\
        & =
        \max \{ \norm{T}_\diamond, \norm{S}_\diamond \}
\end{flalign*}
From this last result it follows, similarly to the case of $\Ban$, that the
binary coproducts of $\Kar{\CPTP}$ are enriched over $\Met$'s Cartesian
structure. This thus establishes a $\Met$-enriched co-Cartesian symmetric
monoidal category. Unfortunately, it is unclear whether this category is
autonomous, a core requisite for being a model of our calculus (recall
Definition~\ref{defn:model}). So in order to overcome the issue, we apply our
results from the previous section. More specifically we adopt the category
$[\Kar{\CPTP}^\cop,\Met]_{\fp}$, clearly a model of our calculus (recall
Section~\ref{sec:models}), and the respective Yoneda embedding $\Kar{\CPTP} \to
[\Kar{\CPTP}^\cop,\Met]_{\fp}$ which preserves not only the symmetric monoidal
structure of $\Kar{\CPTP}$ but also binary coproducts.

As in the probabilistic case, we could now develop a metric $\lambda$-theory
for reasoning about Cauchy sequences of \eg\ \emph{quantum} random walks. But
since this would be analogous in spirit to our illustration in the preceding
subsection, we refrain from doing so.

\section{Conclusions and future work}
\label{sec:conc}

We extended quantalic linear $\lambda$-calculus~\cite{dahlqvist22,dahlqvist23}
with additive disjunction $\oplus$, and showed that the extension retains
several desirable properties of the original calculus if one works with
\emph{continuous} quantales.  While we did not explore the case of
\emph{additive conjunction} (which categorically corresponds to products) we
expect that it can be added to the calculus in an analogous way. 

Our work opens up several avenues of research. For example, since
the calculus' linearity can be overly restrictive in some programming
paradigms, it would be interesting to merge our results with quantalic
\emph{graded} $\lambda$-calculus introduced in~\cite{dahlqvist23b}.  As another
example, the fact that the gluing construction extends to our quantalic
setting, together with the established reflection,
\[
        \mathscr{L} :  [\catC^\cop,\VCatSe] \longrightarrow [\catC^\cop,\VCatSe]_{\fp}
\]
(recall Section~\ref{sec:models}) makes possible to study `conservativity'
results of our calculus (and similar ones) in the style of
Lafont~\cite{lafont88}. Also, the same reflection can be used for exploring
corresponding `concrete completeness' results~\cite{johnstone02}.  As yet
another example, the rich structure of the (enriched) Karoubi envelope
$\Kar{\CPTP}$ -- which in particular supports so-called splits of
`symmetrisations' -- makes this category a promising candidate for interpreting
graded modalities~\cite{dahlqvist23b,dahlqvist19}. It would be interesting to
explore the connections between such modalities and symmetric subspaces which
have vast applications in quantum information theory~\cite[Chapter
7]{watrous18}.

\bibliographystyle{alpha} 
\bibliography{biblio}

\end{document}